\documentclass[lettersize,journal]{IEEEtran}
\usepackage{amsmath,amsfonts}
\usepackage{amssymb}
\usepackage{algorithmic}
\usepackage{algorithm}
\usepackage{array}
\usepackage[caption=false,font=footnotesize,labelfont=rm,textfont=rm]{subfig}
\usepackage{comment}
\usepackage{float}
\usepackage{textcomp}
\usepackage{stfloats}
\usepackage{url}
\usepackage{verbatim}
\usepackage{graphicx}
\usepackage{subfig} 
\usepackage{subfloat}
\usepackage{cite}
\usepackage{hyperref}
\usepackage{color}
\usepackage{xcolor}
\usepackage[breakable]{tcolorbox}
\newtheorem{lemma}{Lemma}
\newtheorem{theorem}{Theorem}
\newtheorem{corollary}{Corollary}
\newtheorem{definition}{Definition}

\newenvironment{proof}{\par\emph{Proof: }}{\hfill$\blacksquare$\newline}

\newcommand*{\USC}{\mathop{}\!\mathrm{USC}}
\newcommand*{\LSC}{\mathop{}\!\mathrm{LSC}}
\newcommand*{\BUC}{\mathop{}\!\mathrm{BUC}}
\begin{document}

\title{Beyond Gaussian Assumptions: A General Fractional HJB Control Framework for 
Lévy-Driven Heavy-Tailed Channels in 6G}

\author{Mengqi Li,~\IEEEmembership{Student Member,~IEEE}, Lixin Li,~\IEEEmembership{Member,~IEEE}, Wensheng Lin,~\IEEEmembership{Member,~IEEE}, \\
Zhu Han,~\IEEEmembership{Fellow,~IEEE}, and Tamer Başar,~\IEEEmembership{Life Fellow,~IEEE}
\thanks{This paper has been accepted for publication in the IEEE Transactions on Wireless Communications with DOI: \href{https://doi.org/10.1109/TWC.2025.3631903}{10.1109/TWC.2025.3631903}.

Corresponding authors: Lixin Li, Wensheng Lin.
	
	This work was supported in part by National Natural Science Foundation of China under Grants 62571450 and 62101450, in part by Key Research and Development Program of Shaanxi under Grant 2025CY-YBXM-043, in part by the Young Elite Scientists Sponsorship Program by the China Association for Science and Technology under Grant 2022QNRC001, in part by Aeronautical Science Foundation of China under Grants 2022Z021053001 and 2023Z071053007, in part by the Open Fund of Intelligent Control Laboratory, in part by the Open Fund of Key Laboratory of Radio Spectrum Testing Technology (The State Radio\_monitoring\_center Testing Center), Ministry of Industry and Information Technology, in part by AFOSR Grant FA9550-24-1-0152, in part by NSF ECCS-2302469, Amazon and Japan Science and Technology Agency (JST) Adopting Sustainable Partnerships for Innovative Research Ecosystem (ASPIRE) JPMJAP2326.}
\thanks{Mengqi Li, Lixin Li, Wensheng Lin are
with the School of Electronics and Information, Northwestern Polytechnical
University, Xi’an, Shaanxi 710129, China (e-mail: mengqili@mail.nwpu.edu.cn; lilixin@nwpu.edu.cn; linwest@nwpu.edu.cn).}
\thanks{Zhu Han is with the Department of Electrical and Computer Engineering at the University of Houston, Houston, TX 77004 USA, and also with the Department of Computer Science and Engineering, Kyung Hee University, Seoul, South Korea, 446-701 (e-mail: hanzhu22@gmail.com).}
\thanks{Tamer Başar is with the Department of Electrical and Computer Engineering, University of Illinois Urbana-Champaign, Urbana, IL 61801 USA (e-mail: basar1@illinois.edu).}
}

\markboth{}%
{}


\maketitle

\begin{abstract}
Emerging 6G wireless systems suffer severe performance degradation in challenging environments like high-speed trains traversing dense urban corridors and Unmanned Aerial Vehicles (UAVs) links over mountainous terrain. These scenarios exhibit non-Gaussian, non-stationary channels with heavy-tailed fading and abrupt signal fluctuations. To address these challenges, this paper proposes a novel wireless channel model based on symmetric $\alpha$-stable Lévy processes, thereby enabling continuous-time state-space characterization of both long-term and short-term fading. Building on this model, a generalized optimal control framework is developed via a fractional Hamilton-Jacobi-Bellman (HJB) equation that incorporates the Riesz fractional operator to capture non-local spatial effects and memory-dependent dynamics. The existence and uniqueness of viscosity solutions to the fractional HJB equation are rigorously established, thus ensuring the theoretical validity of the proposed control formulation. Numerical simulations conducted in a multi-cell, multi-user downlink setting demonstrate the effectiveness of the fractional HJB-based strategy in optimizing transmission power under heavy-tailed co-channel and multi-user interference.
\end{abstract}


\section{Introduction}
The sixth-generation (6G) wireless communication network aims to support hyper-reliable low-latency communications (HRLLC) across diverse deployment scenarios~\cite{r1,r2,r3}. However, the widespread deployment of high-mobility platforms—including aircraft, unmanned aerial vehicles (UAVs), and high-speed trains—introduces fundamental challenges to accurate analytical channel modeling and robust control design. Wireless links in such environments experience heavy-tailed fading and abrupt signal variations due to heterogeneous propagation conditions and impulsive interference~\cite{R1,Basar}, violating the classical assumptions of Gaussianity and stationarity. As a result, traditional wireless models become inadequate for characterizing wireless channel dynamics in 6G HRLLC scenarios~\cite{r4}.

Traditional fading models decompose channel variations into long-term and short-term components. Long-term fading, induced by terrain and structural obstructions, is typically modeled using log-normal distributions, while short-term fading due to multipath propagation is often characterized by Rayleigh or Rician models~\cite{FadingChannel}. To address temporal variability, Charalambous and Menemenlis introduced stochastic differential equation-based models, where long-term fading follows a mean-reverting Ornstein-Uhlenbeck (O-U) process and short-term fading is represented by envelope dynamics governed by square-root diffusions~\cite{r32,r33}. Although effective in quasi-static environments, these models struggle to capture the impulsive, heavy-tailed behaviors observed in highly dynamic and mobile scenarios, thus potentially limiting their applicability under non-Gaussian, non-stationary interference~\cite{r11,r7}.

To address these limitations, a more sophisticated stochastic framework is required. The noise and interference in challenging wireless environments are often impulsive rather than purely thermal, originating from sources like man-made electronics, atmospheric discharges, or sudden multipath configuration changes. This impulsiveness means the underlying statistical distribution is heavy-tailed—that is, the probability of observing random events of very large amplitude is significantly higher than what a Gaussian distribution would predict. The $\alpha$-stable distribution is a generalization of the Gaussian distribution, characterized by a stability index $\alpha\in(0,2]$ that governs the heaviness of its tails. While $\alpha=2$ recovers the Gaussian case, values of $\alpha<2$ generate distributions with infinite variance, making them highly suitable for modeling impulsive phenomena ~\cite{R2,R3}. However, a static distribution is insufficient because channel noise evolves over time. A Lévy process provides the necessary temporal dynamics~\cite{R4}. 

By driving a Lévy process with a Symmetric $\alpha$-stable (S$\alpha$S) distribution, we obtain an S$\alpha$S Lévy process, which models a system that exhibits not only heavy-tailed characteristics at any given moment but also evolves through sudden, discontinuous jumps over time. This makes it an exceptionally well-suited framework for representing both the statistical properties and the dynamic behavior of impulsive noise in wireless channels. In the wireless domain, Liu et al. proposed a GLRT-based detector for target detection in sub-Gaussian symmetric $\alpha$-stable clutter by expressing the detection statistic using Fox's $H$-function~\cite{GLRT}, while Zhang et al. developed a transmit antenna identification algorithm for MIMO cognitive radio systems under $\alpha$-stable noise using fractional lower-order statistics~\cite{MIMO}. Building on this foundation, this study proposes a wireless channel model based on S$\alpha$S Lévy processes, aiming to extend the dynamics into a continuous-time state-space formulation and jointly capture long-term and short-term fading under non-stationary propagation conditions.

While such modeling provides an accurate statistical description of non-Gaussian fading, it remains insufficient to ensure stable communication performance under interference and uncertainty. Abrupt fluctuations caused by $\alpha$-stable jumps can degrade system stability, particularly in multi-user networks with shared spectral resources. These challenges necessitate a robust control framework capable of real-time adaptation through strategies such as power allocation, scheduling, and beamforming~\cite{R7,R8,r17}. Dynamic programming (DP) offers a structured methodology for such sequential decision problems. For continuous-time systems, the cornerstone of DP is the Hamilton-Jacobi-Bellman (HJB) equation, a non-linear partial differential equation that provides a necessary and sufficient condition for optimality~\cite{r18,r19,r20}. The solution to the HJB equation is the value function, which quantifies the optimal expected cost-to-go from any given state and time. By solving for this value function, one can derive the optimal control policy for the entire system. However, the structure of the classical HJB equation, specifically its reliance on second-order differential operators, is fundamentally tied to the assumption of underlying Gaussian noise. This makes it unsuitable for accounting for the non-local and infinite-variance behavior induced by Lévy processes~\cite{Jump}.

To bridge this gap, this paper introduces a fractional HJB equation, replacing the standard Laplacian with a Riesz fractional derivative. This change is a direct consequence of the underlying stochastic model. In optimal control, the HJB equation includes a term derived from the infinitesimal generator of the stochastic process driving the system. The standard Laplacian—a second-order, local operator—is the well-known generator for Gaussian/Brownian motion. Correspondingly, for a system driven by a symmetric $\alpha$-stable Lévy process, the generator is the Riesz fractional operator, which is a non-local, integral operator. As the infinitesimal generator of symmetric $\alpha$-stable (S$\alpha$S) motions, the Riesz operator provides a rigorous framework for modeling non-local spatial interactions and power-law memory effects—features intrinsic to $\alpha$-stable Lévy dynamics~\cite{r21,r22,r23}. Its mathematical non-locality provides a formal mechanism to account for the physical possibility of the channel state making large, instantaneous jumps, thereby ensuring the control policy considers such drastic changes in system conditions. This formulation provides a framework for designing control strategies aimed at optimizing performance in wireless systems subject to heavy-tailed and discontinuous stochastic variations.

However, the inclusion of fractional operators and Lévy-driven jumps poses analytical challenges, as the resulting value function may lack classical differentiability~\cite{r24}. To address this, the framework of viscosity solutions is adopted, originally introduced by Lions for fully nonlinear partial differential equations~\cite{r25,r26,r27}. This approach ensures mathematical well-posedness through two key techniques: the comparison principle~\cite{r28}, which establishes the ordering of subsolutions and supersolutions, and the Perron’s method~\cite{r29}, which constructs a unique solution as the supremum of admissible subsolutions. These tools collectively guarantee the existence and uniqueness of solutions to the fractional HJB equation and validate the derived control strategies under non-Gaussian channel dynamics.

In summary, this study addresses the critical need for advanced modeling and control in 6G by presenting a unified framework with a clear logical progression. First, a novel Lévy-driven channel model is introduced that is well-suited to capturing the heavy-tailed, non-stationary characteristics of challenging wireless environments. Second, based on this physically-motivated model, the corresponding fractional HJB control framework is derived, where the Riesz fractional operator arises as the appropriate generator for the jump dynamics. Finally, viscosity solution theory is employed to ensure the resulting optimization problem is mathematically robust and well-posed. These contributions provide a unified approach for modeling and controlling wireless systems under stochastic, non-Gaussian fading, and establish a foundation for future research in fractional-order control and wireless system optimization.

This study makes the following contributions:
\begin{itemize}
    \item A novel Lévy-driven wireless channel model is developed based on S$\alpha$S processes, which extends classical fading models into a continuous-time state-space framework and accurately captures both long-term and short-term fading under non-Gaussian and non-stationary propagation conditions.
    
    \item A generalized fractional HJB control framework is formulated, incorporating Riesz fractional operators to model non-local spatial interactions and memory-dependent dynamics induced by impulsive Lévy-driven fading. This framework provides a robust foundation for optimal control in stochastic, heavy-tailed wireless environments.
    
    \item The existence and uniqueness of viscosity solutions to the fractional HJB equation are rigorously established, leveraging comparison principles and the Perron’s method. This guarantees the mathematical well-posedness of the control problem under fractional-order and non-smooth dynamics.

    \item A numerical simulation is conducted to validate the proposed fractional HJB-based control strategy in a same-frequency, multi-base-station, multi-user scenario, demonstrating effective transmission power adaptation and interference mitigation under heavy-tailed channel conditions.
\end{itemize}

The remainder of this paper is organized as follows:  Section~\ref{S2} introduces the system model, including the network architecture, signal definitions, and statistical assumptions underpinning the analysis. Section~\ref{S3} develops a stochastic wireless channel model that captures both long-term and short-term fading using S$\alpha$S Lévy processes and fractional-order dynamics. Section~\ref{S4} presents a general optimal control framework for non-Gaussian stochastic systems, deriving the corresponding fractional HJB equation and proving the existence and uniqueness of its viscosity solution. Section~\ref{S5} implements the proposed framework in a downlink power control scenario, validating the effectiveness of the model through numerical simulations and performance evaluations. Finally, Section~\ref{S6} concludes this paper.

\textit{Notations:} The mathematical concepts and scope of the key symbols used in this article are shown in Table \ref{tab:notation}. 
\begin{table}[t]
		\centering
		\caption{Summary of Notations}
		\begin{tabular}{cc}
			\hline
			\textbf{Symbol} & \textbf{Description} \\
			\hline
			$\mathbb{R}^+$ & Set of non-negative real numbers \\
			$\mathbb{R}$ & Set of all real numbers \\
            $j$ & imaginary unit\\
            $\alpha$ & Stability index of the Lévy process, where $\alpha \in (0,2]$ \\ 
			$L^\alpha$ & S$\alpha$S Lévy process with heavy tails and jumps \\
            $\operatorname{sgn}(\cdot)$ & Sign function\\
			$\mathrm{d}[\cdot]^c_t$ & Quadratic variation of the continuous local martingale part \\
			$\mathbb{E}(\cdot)$ & Expectation operator for computing expected values \\
			$\mathcal{A}$ & Infinitesimal generator of a stochastic process \\
			$\mathbf{1}$ & Indicator function: equals 1 if condition holds, 0 otherwise \\
			$\nabla^\alpha$ & Generalized Riesz fractional derivative with non-locality \\
			$(\cdot)^\top$ & Transpose of a vector or a matrix \\
			$\USC$ & Space of upper semicontinuous functions \\
			$\LSC$ & Space of lower semicontinuous functions \\
			$\BUC$ & Space of bounded uniformly continuous functions \\
			$C^1$ & Continuously differentiable with respect to time\\
            $C^2$ & Twice continuously differentiable in space\\
			\hline
		\end{tabular}
		\label{tab:notation}
	\end{table}

\section{System Model}\label{S2}
\begin{figure*}[!t]
    \centering
    \includegraphics[width=0.65\linewidth]{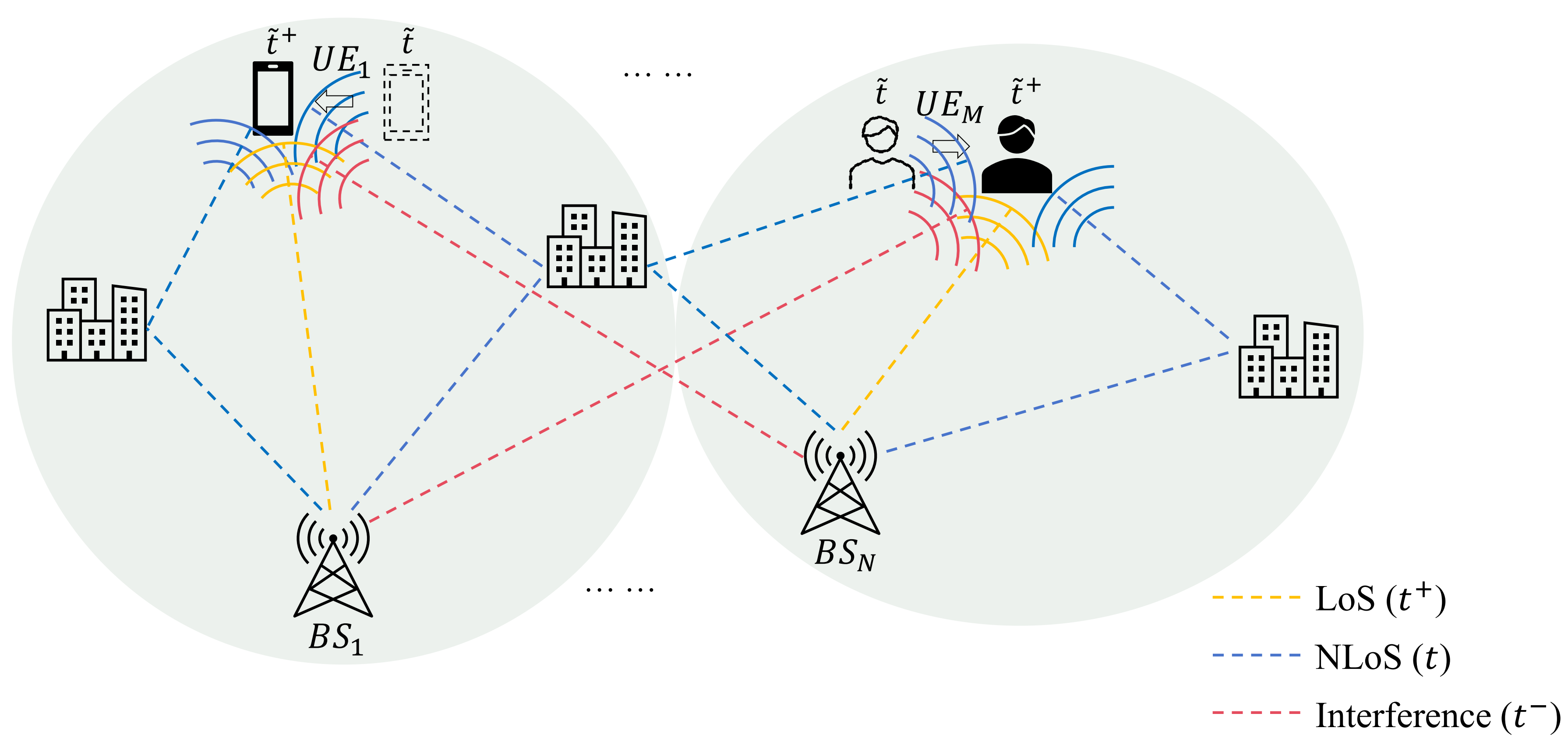}
    \caption{Illustration of signal components, transmitted at different moments, arriving simultaneously at the UE at moment $\tilde{t}^+$.}
    \label{fig}
\end{figure*}
Considering a wireless network comprising base stations (BSs) $i~(1 \leq i \leq N)$, each serving user equipments (UEs) $\ell~(1 \leq \ell \leq M)$ over a shared frequency band, mutual interference is introduced due to co-channel transmissions. For a mobile UE, the signal received at any given time is a superposition of components that have traversed paths of different lengths and were thus transmitted at different moments.

Fig. \ref{fig} illustrates this physical reality. It depicts the composition of the total signal arriving at the UE at a specific moment, denoted $\tilde{t}^+$, where $\tilde{t}^+ = t^+ + \varrho$ and $\varrho$ represents LoS transmission delay (solid black user icon), as it moves from its position at a prior moment $\tilde{t}$ (dashed white user icon). Due to varying propagation delays, these simultaneously arriving signal components were transmitted at different times: (i) Direct line-of-sight (LoS) Path (yellow): The LoS signal, having the shortest propagation path, was transmitted from the serving BS at moment $t^+$, and arrived at the UE at moment $\tilde{t}^+$. (ii) Multipath non-line-of-sight (NLoS) Components (blue): The NLoS components travel longer paths via reflection and scattering off obstacles like buildings. Fig. \ref{fig} depicts a representative NLoS path, but in reality, the signal arrives via numerous such paths. The constructive and destructive superposition of these multiple, time-delayed signal replicas is the origin of multipath fading. These components were transmitted at a slightly earlier moment $t$, and arrived at the UE at the same moment $\tilde{t}^+$. (iii) Interference Components (red): Interference signals from neighboring BSs, which typically travel the longest paths, were transmitted at an even earlier moment $t^-$, and arrived at the UE at the same moment $\tilde{t}^+$.

Crucially, in the challenging high-mobility environments considered, each of these signal components is subject to heavy-tailed fluctuations from impulsive interference and complex, non-stationary scattering. This violates the assumptions of classical models that rely on the Central Limit Theorem, causing the total received power at the UE to exhibit abrupt, non-Gaussian jumps rather than conventional Rayleigh or Rician fading. In the state-space representation, long-term fading is modeled by a stochastic differential equation (SDE) driven by S$\alpha$S Lévy noise. This SDE captures the mean-reverting dynamics of slow shadowing variations on individual propagation paths. Rapid multipath fluctuations over very short time scales are modeled by superimposing $K$ independent plane waves with different amplitudes and phases.

\subsection{Fading Structure and Power Evolution}
Let \( p_{\text{in}}(t) \) denote the transmit power of the BS and \( p_{\text{out}}(t) \) the received power at the UE served by the BS. To capture multi-scale propagation effects, the wireless channel is modeled as a concatenation of long-term fading \( \beta(t) \) and short-term fading \( \chi(t) \), applied sequentially to the transmitted power. Specifically, it indicates that the transmitted signal is affected by both path loss and long-term fading effects caused by object reflection before reaching the local area of the receiver. It then scatters in the local area of the receiver and recombines at the receiver, corresponding to the short-term fading effect in the local area. The system structure is summarized as shown in Fig.~\ref{flow}.
\begin{figure}[h]
    \centering
    \includegraphics[width=0.95\linewidth]{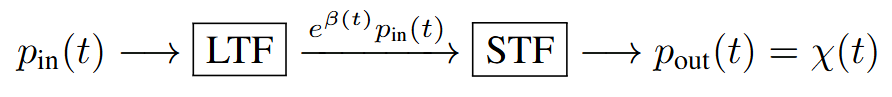}
    \caption{System Flow Chart.}
    \label{flow}
\end{figure}
This representation defines the dynamic transfer relationship between transmitted and received power, thereby enabling a continuous-time modeling approach that jointly captures large-scale attenuation and small-scale fluctuations.

\subsection{Classical Fading Mechanisms}
\subsubsection{Long-Term Fading}
Long-term fading refers to signal variations caused by large-scale path loss and shadowing. Classical models, such as log-normal fading, describe its statistical distribution but not its temporal dynamics. In reality, channel conditions vary continuously due to terrain, obstructions, and mobility. Previous works, such as Charalambous et al. \cite{r32}, have used Wiener processes to model long-term signal envelope dynamics. Let the long-term fading channel gain be given by $e^{\beta(t)}$ where $\beta(t)$ is a continuous-time stochastic process following the first order linear SDE
\begin{equation}
\mathrm{d}\beta(t) = -a (\beta(t) + b) \mathrm{d}t + \sigma_\beta \mathrm{d}W_\beta(t), 
\end{equation}
where $t\in\mathbb{R}^+$. The drift term $-a (\beta(t) + b)$ enforces mean reversion toward the equilibrium value $-b~(b>0)$, with convergence rate governed by $a > 0$. The diffusion term $\sigma_\beta \mathrm{d}W_\beta(t)$ captures random fluctuations, where $W_\beta(t)$ is a standard Wiener process and $\sigma_\beta$ denotes the noise intensity. The process is initialized at $\beta(0) = \beta_0$. Macroscopically, because the effective channel power gain $e^{\beta(t)}$ follows a log-normal distribution at any given time $t$, this model is consistent with the traditional log-normal model.

\subsubsection{Short-Term Fading}
Traditional short-term fading models represent the received signal as a superposition of multiple scattered plane waves. The in-phase and quadrature components are typically modeled as independent, zero-mean Gaussian random variables, leading to a Rayleigh-distributed envelope under the Central Limit Theorem (CLT). This formulation assumes a large number of independent and identically distributed (i.i.d.) components with finite variance, making it suitable for stationary and Gaussian environments.

Under local multipath scattering, the transmitted signal is decomposed into $K$ plane wave components, each arriving at the receiver with distinct amplitudes and phase shifts. The $k$-th component is given by, for $t \in \mathbb{R}^+$,
\begin{equation}
    E_k(t)=I_k(t) \cos \left(\omega_c t\right)-Q_k(t) \sin \left(\omega_c t\right),
\end{equation}
where $\omega_c$ is the carrier frequency and $I_k(t)$ and $Q_k(t)$ are the in-phase and quadrature components, defined respectively as
\begin{equation}
    \begin{aligned}
& I_k(t)=\sum_{m=1}^L r_{k, m} \cos \left(\omega_{k, m} t+\theta_{k, m}\right),  \\
& Q_k(t)=\sum_{m=1}^L r_{k, m} \sin \left(\omega_{k, m} t+\theta_{k, m}\right), 
\end{aligned}
\end{equation}
where $r_{k, m}$, $\omega_{k, m}$ and $\theta_{k, m}$ represent the amplitude, Doppler shift, and phase of the $(k,m)$-th wave component, respectively. By a transformation of random variables, the signal envelope at the receiver, $r_k(t)=\sqrt{I_k(t)^2 +Q_k(t)^2}$ is Rayleigh distributed. The received signal power, $r_k(t)^2$ is exponentially distributed, again by a transformation of random variables.

\subsection{$\alpha$-Stable Distributions}
To model impulsive fading and non-Gaussian effects, the $\alpha$-stable distribution is adopted as the underlying noise model. Specifically, the $r_{k,m}$ and $\theta_{k,m}$ mentioned above are modeled as independent random variables drawn from $\alpha$-stable distributions, thereby introducing heavy-tailed behaviors into the received signal components. The characteristic function of the $\alpha$-stable distribution \cite{Levy} is given by

\begin{equation}
\mathbb{E}\left[e^{j k L^\alpha}\right]=\exp \left(j \mu k-\gamma|k|^\alpha\left(1-j \beta \operatorname{sgn}(k) \tan \frac{\pi \alpha}{2}\right)\right),
\end{equation}
where $\alpha \in (0,2]$ is the stability index controlling tail heaviness and impulsiveness, with $\alpha = 2$ corresponding to the Gaussian case and $\alpha < 2$ yielding infinite variance. The skewness is governed by $\beta \in [-1,1]$, where $\beta = 0$ indicates symmetry. The scale parameter $\gamma > 0$ determines dispersion, and $\mu \in \mathbb{R}$ specifies the location. 

In many practical fading scenarios, particularly those involving aggregated multipath effects, signal components can reasonably be assumed to follow symmetric distributions. Therefore, the S$\alpha$S Lévy processes are adopted in the subsequent modeling. The corresponding characteristic function simplifies to
\begin{equation}
\mathbb{E}\left[ e^{j k L^\alpha} \right] = \exp\{j \mu k-\gamma |k|^\alpha\}, \quad 1 < \alpha \leq 2. 
\end{equation}
This work focuses on the range of $\alpha\in(1,2]$. This restriction is standard in many applications as it ensures that the process has a finite mean, a necessary condition for the expectation in the cost function to be well-defined. This stochastic framework serves as the foundation for modeling signal power dynamics and formulating optimal control strategies under non-Gaussian, interference-limited, and time-varying conditions. 

\section{Wireless Channel Modeling with Lévy-driven Fractional Dynamics}\label{S3}
Modern wireless channel modeling faces significant challenges in reconciling multi-scale propagation dynamics with non-Gaussian statistical characteristics. Although effective under static or mildly changing conditions, traditional models suffer from three major limitations: (i) discrete-scale decomposition fails to capture cross-scale interactions under dynamic conditions; (ii) Gaussian noise assumptions cannot represent impulsive and heavy-tailed phenomena common in urban or obstructed environments; and (iii) discrete-time formulations obscure the inherently continuous nature of signal evolution.

To overcome these limitations, a novel continuous-time state-space channel model is proposed, which unifies macroscopic path loss, microscopic fading, and Lévy-driven perturbations. This unified model accurately represents the relationship between the transmitted input signal power and the received output signal power. In particular, the incorporation of Lévy-driven noise components enables the modeling of abrupt changes, jump discontinuities, and heavy-tailed propagation effects, thereby enhancing the model's ability to describe complex, time-varying communication environments.

\subsection{Non-Gaussian Long-term Fading}
To model large-scale fading under non-Gaussian and non-stationary conditions, the long-term fading process $\beta(t)$ is formulated as an S$\alpha$S Lévy-driven O-U process
\begin{equation} \label{long}
\mathrm{d}\beta(t) = -a (\beta(t) + b) \mathrm{d}t + \sigma_\beta \mathrm{d}L_\beta^\alpha(t), \, \beta(0) = \beta_0, 
\end{equation}
where $t\in\mathbb{R}^+$. $L_\beta^\alpha(t)$ is a S$\alpha$S Lévy process, and $\sigma_\beta$ controls the intensity of jump-induced fluctuations. This formulation captures both smooth and impulsive variations in long-term channel behavior, extending classical models to more accurately reflect real-world fading dynamics. This approach builds upon the established work of using Gaussian-driven O-U processes for dynamic long-term fading \cite{r32} and extends it to the non-Gaussian case to better reflect real-world fading dynamics.

\subsection{Non-Gaussian Short-term Fading}
Although Gaussian-based models capture fixed-time statistics, they fail to represent temporal correlations and abrupt signal fluctuations. In practice, wireless channels often exhibit heavy-tailed behavior and amplitude jumps due to obstructions or environmental dynamics, where the finite-variance assumption and classical CLT no longer apply. To address these effects, the classical framework is extended using S$\alpha$S Lévy processes, which rely on the Generalized CLT (GCLT) to model non-Gaussian dynamics under infinite-variance conditions.

\subsubsection{In-Phase and Quadrature Components}

In conventional models, the processes $I(t)$ and $Q(t)$ are often represented using SDEs driven by Brownian motion \cite{r33}. To incorporate both the Doppler effect and heavy-tailed jump phenomena, these are replaced with S$\alpha$S Lévy processes, yielding the following SDEs
\begin{equation} \label{IQ}
\begin{aligned} 
\mathrm{d}I_k(t) &= -\frac{1}{2} \kappa I_k(t) \mathrm{d}t + \frac{1}{2} \sigma \mathrm{d}L_{I_k}^\alpha(t), \quad I_k(0)=(I_k)_0,\\ 
\mathrm{d}Q_k(t) &= -\frac{1}{2} \kappa Q_k(t) \mathrm{d}t + \frac{1}{2} \sigma \mathrm{d}L_{Q_k}^\alpha(t), \quad Q_k(0)=(Q_k)_0,
\end{aligned} 
\end{equation}
where $\kappa > 0$ is the damping coefficient controlling the rate of decay, and $\sigma > 0$ governs the intensity of fluctuations and jumps. $L_{I_k}^\alpha(t)$ and $L_{Q_k}^\alpha(t)$ are independent S$\alpha$S Lévy processes. The initial values $(I_k)_0$ and $(Q_k)_0$ define the starting states.

\subsubsection{Signal Envelope and Power}

The instantaneous power $\chi_k(t)$ and the signal envelope $r_k(t)$ are defined as
\begin{equation}
\chi_k(t) = I_k(t)^2 + Q_k(t)^2, \quad r_k(t) = \sqrt{\chi_k(t)}. 
\end{equation}
Compared with classical Rayleigh based models, this model is affected by the heavy tails and jumps of $I_k(t)$ and $Q_k(t)$ in the envelope $r_k(t)$, resulting in a higher probability of extreme fading events.

\subsubsection{Dynamics of the Instantaneous Power}
Before delving into the detailed derivation, we provide some necessary lemmas and properties.
\begin{lemma}[Generalized Itô Formula for Jump Semimartingales \cite{SDE}]\label{lemma:ito-general}
Let $X(t)$ be a real-valued càdlàg semimartingale with jumps, and let $f \in C^2(\mathbb{R})$ be a twice continuously differentiable function. Then, the process $Y(t) = f(X(t))$ satisfies
\begin{align}
\mathrm{d}f(X_t) &= f'(X_{t-})\,\mathrm{d}X(t) + \frac{1}{2}f''(X_{t-})\,\mathrm{d}[X]^c_t \notag\\
&\quad + \sum_{0 < s \le t} \left[ f(X_s) - f(X_{s-}) - f'(X_{s-})\,\Delta X_s \right],
\end{align}
where $X_{t-}$ denotes the left-limit of $X$ at time $t$ (well-defined since $X$ is càdlàg). $\Delta X_s = X_s - X_{s-}$ is the jump size of $X$ at time $s$. $[X]^c_t$ is the quadratic variation of the continuous local martingale part of $X$.
\end{lemma}

\begin{lemma}[Lévy–Itô Decomposition of a Pure Jump Process \cite{r43}]\label{lemma:levyito}
Let $L(t)$ be a pure-jump Lévy process with Lévy measure $\nu$ and no Brownian motion part. For any truncation threshold $\epsilon > 0$, the process can be decomposed into a compensated small-jump component and a drift-compensated jump integral as follows
\begin{equation}
L(t) = \int_0^t \int_{|z| \leq \epsilon} z\, \tilde{N}(\mathrm{d}s, \mathrm{d}z)
+ \int_0^t \int_{|z| > \epsilon} z\, N(\mathrm{d}s, \mathrm{d}z),
\end{equation}
where $N(\mathrm{d}s, \mathrm{d}z)$ is the Poisson random measure associated with the jumps of $L(t)$. $\tilde{N}(\mathrm{d}s, \mathrm{d}z) = N(\mathrm{d}s, \mathrm{d}z) - \nu(\mathrm{d}z)\mathrm{d}s$ is the compensated Poisson measure (for small jumps), where $\nu(\mathrm{d}z)$ is the Lévy measure satisfying $\int_{\mathbb{R}} \min(1, z^2)\nu(\mathrm{d}z) < \infty$.

In particular, for any predictable integrand $h(t,z)$, the jump contribution to the quadratic variation satisfies
\begin{equation}
\sum_{0 < s \le t} h(s, \Delta L(s))^2
= \int_0^t \int_{\mathbb{R}} h(s,z)^2\, N(\mathrm{d}s, \mathrm{d}z).
\end{equation}
\end{lemma}

\begin{lemma}[Closure of S$\alpha$S Distributions under Scaling and Linear Combinations \cite{Samorodnitsky}]\label{lemma:stable-closure}
Let $X$ and $Y$ be independent symmetric $\alpha$-stable random variables with $1 < \alpha \le 2$. Then:
\begin{enumerate}
    \item For any real constant $c \in \mathbb{R}$, $cX$ is also $\alpha$-stable with scale $|c|$.
    \item For any $c_1, c_2 \in \mathbb{R}$, the linear combination $Z = c_1 X + c_2 Y$ is also symmetric $\alpha$-stable with the same index $\alpha$ and scale parameter
    \begin{equation}
        \gamma_Z = \left( |c_1|^\alpha + |c_2|^\alpha \right)^{1/\alpha}.
    \end{equation}
\end{enumerate}
\end{lemma}

Initially, by applying the two-dimensional case of Lemma \ref{lemma:ito-general}, $\chi_k(t) = f(I_k(t),Q_k(t))=I_k^2(t) + Q_k^2(t)$ yields
\begin{equation} 
\begin{aligned}
\mathrm{d} \chi_k(t)= & 2 I_k(t-) \mathrm{d} I_k(t)+2 Q_k(t-) \mathrm{d} Q_k(t)+\mathrm{d}\left[I_k, I_k\right]_t^c\\
&+\mathrm{d}\left[Q_k, Q_k\right]_t^c +\sum_{0<s \leq t}\left\{[I_k(s)]^2+[Q_k(s)]^2\right.\\
&- [I_k(s-)]^2-[Q_k(s-)]^2\\
&\left.-2 I_k(s-) \Delta I_k(s)-2 Q_k(s-) \Delta Q_k(s)\right\}.
\end{aligned}
\end{equation}
Since $[I_k(s)]^2-[I_k(s-)]^2=2I_k(s-)\Delta I_k(s)+[\Delta I_k(s)]^2$, (and similarly for $Q_k$), the jump term becomes
\begin{equation}
\sum_{0<s \leq t}\left\{[\Delta I_k(s)]^2+[\Delta Q_k(s)]^2\right\} .
\end{equation}
For the $\alpha$-stable regime with $1<\alpha<2$, the process typically has no continuous Gaussian component, implying $\mathrm{d}\left[I_k, I_k\right]_t^c=\mathrm{d}\left[Q_k, Q_k\right]_t^c=0$. In the pure Gaussian case $\alpha=2$, there would be a continuous Brownian motion part. Thus, the dynamics of $\chi_k(t)$ are described as
\begin{equation}\label{chi}
\begin{aligned}
\mathrm{d} \chi_k(t)= & 2 I_k(t-) \mathrm{d} I_k(t)+2 Q_k(t-) \mathrm{d} Q_k(t)\\
&+\sum_{0<s \leq t}\left\{[\Delta I_k(s)]^2+[\Delta Q_k(s)]^2\right\}.
\end{aligned}
\end{equation}

Substituting \eqref{IQ} into \eqref{chi}, the evolution of $\chi_k(t)$ becomes
\begin{align}
\mathrm{d} \chi_k(t)=&-\kappa \chi_k(t-) \mathrm{d} t+\sum_{0<s \leq t}\left(\left[\Delta I_k(s)\right]^2+\left[\Delta Q_k(s)\right]^2\right)\notag\\
&+
\sigma\left[I_k(t-) \mathrm{d} L_{I_k}^\alpha(t)+Q_k(t-) \mathrm{d} L_{Q_k}^\alpha(t)\right] .
\end{align}
For $\alpha < 2$, the infinite variance property implies
\begin{equation}
\int_{\mathbb{R}} z^2 \nu(\mathrm{d} z)=\infty,
\end{equation}
which precludes incorporating the jump-squared term into the drift, as the integral diverges. To address this, the Lévy measure is truncated by restricting jumps to \( |z| \leq \epsilon \), yielding the modified measure \( \nu_\epsilon(\mathrm{d}z) = \nu(\mathrm{d}z) \mathbf{1}_{|z| \leq \epsilon} \), which satisfies
\begin{equation}
\int_{|z| \leq \epsilon} z^2 \nu_\epsilon(\mathrm{d} z)<\infty.
\end{equation}
After truncation, process $L_{I_k}^{\alpha, \epsilon}(t)=\int_0^t \int_{|z| \leq \epsilon} z \tilde{N}_I(\mathrm{d} s, \mathrm{d} z)$ and similar $L_{Q_k}^{\alpha, \epsilon}(t)$ become finite-variance Lévy processes, where $\tilde{N}_I$ and $\tilde{N}_Q$ are compensated Poisson random measures ensuring zero mean. Using Lemma \ref{lemma:levyito}, the contribution of the jump process becomes
\begin{align}
\sum_{0<s \leq t}&\left[\left(\Delta I_k(s)\right)^2+\left(\Delta Q_k(s)\right)^2\right]\notag\\
&=\left(\frac{\sigma}{2}\right)^2 \int_0^t \int_{|z| \leq \epsilon} z^2\left[N_I(\mathrm{d} s, \mathrm{d} z)+N_Q(\mathrm{d} s, \mathrm{d} z)\right]\notag\\
&+2\left(\frac{\sigma}{2}\right)^2 \int_0^t \int_{|z| \leq \epsilon} z^2 \nu_\epsilon(\mathrm{d} z) \mathrm{d} s,
\end{align}
where $N_I(\mathrm{d} s, \mathrm{d} z)=\tilde{N}_I(\mathrm{d} s, \mathrm{d} z)+\nu_\epsilon(\mathrm{d} z) \mathrm{d} s$, $N_Q(\mathrm{d} s, \mathrm{d} z)=\tilde{N}_Q(\mathrm{d} s, \mathrm{d} z)+\nu_\epsilon(\mathrm{d} z) \mathrm{d} s$. 
Since $I_k$ and $Q_k$ processes are symmetric and independent of identical Lévy measures, the sum of jump processes simplifies to
\begin{equation}
\sum_{0<s \leq t}\left[\left(\Delta I_k(s)\right)^2+\left(\Delta Q_k(s)\right)^2\right]=\frac{\sigma^2}{2} \int_0^t C_\epsilon \mathrm{d} s+M(t),
\end{equation}
with 
$$ C_\epsilon:=\int_{|z| \leq \epsilon} z^2 \nu_\epsilon(\mathrm{d} z), $$
$$ M(t):=\left(\frac{\sigma}{2}\right)^2 \int_0^t \int_{|z| \leq \epsilon} z^2\left(\tilde{N}_I+\tilde{N}_Q\right) \mathrm{d} z \mathrm{d} s, $$
where $C_\epsilon$ aggregates the deterministic second-order moment contributions of all Lévy jumps with amplitudes satisfying $|z|\leq\epsilon$, thereby fixing the drift intensity. Specifically, it is equivalent to converting indistinguishable small burst fading into a deterministic noise floor, increasing the amplitude of the drift term in power evolution SDE. The process $M(t)$ is a square-integrable martingale with zero expectation, capturing the random fluctuations of the cumulative squared jumps around their mean value specified by $C_\epsilon$. Physically, it describes the instantaneous power random fluctuations introduced by truncated small jumps, which cause unpredictable but zero mean disturbances to the evolution of $\chi_k (t)$.

Define the normalized Lévy process $\tilde{L}_k^\alpha(t)$,
\begin{equation}
\mathrm{d} \tilde{L}_k^\alpha(t)=\frac{I_k(t-) \mathrm{d} L_{I_k}^\alpha(t)+Q_k(t-) \mathrm{d} L_{Q_k}^\alpha(t)}{\sqrt{\chi_k(t-)}}.
\end{equation}
Let $L_{I_k}^\alpha(t)$ and $L_{Q_k}^\alpha(t)$ be independent S$\alpha$S Lévy processes. The in-phase and quadrature processes $I_k(t)$ and $Q_k(t)$ satisfy SDEs of the form
\begin{equation}
\mathrm{d} I_k(t)=f\left(I_k(t)\right) \mathrm{d} t+g\left(I_k(t)\right) \mathrm{d} L_{I_k}^\alpha(t),
\end{equation}
and similarly for $Q_k(t)$. Under standard Lipschitz conditions on $f(\cdot)$ and $g(\cdot)$, these SDE-driven processes remain semimartingales. Hence, $I_k(t-)$ and $Q_k(t-)$ are predictable with respect to the filtration generated by $L_{I_k}^\alpha$ and $L_{Q_k}^\alpha$.

At any jump time $\tau$, the process $\tilde{L}_k^\alpha(t)$ admits
\begin{equation}
	\Delta \tilde{L}_k^\alpha(\tau) =
	\frac{I_k(\tau-) \Delta L_{I_k}^\alpha(\tau) + Q_k(\tau-) \Delta L_{Q_k}^\alpha(\tau)}{\sqrt{\chi_k(\tau-)}},
\end{equation}
where both $\Delta L_{I_k}^\alpha(\tau)$ and $\Delta L_{Q_k}^\alpha(\tau)$ are independent symmetric $\alpha$-stable random variables.

Applying Lemma \ref{lemma:stable-closure}, the numerator in $\Delta \tilde{L}_k^\alpha$ remains $\alpha$-stable, and the division by $\sqrt{\chi_k(t-)}$ introduces a predictable scaling. Therefore, $\tilde{L}_k^\alpha(t)$ is a normalized, symmetric $\alpha$-stable semimartingale process with unchanged stability index $\alpha$.

The resulting SDE for $\chi_k(t)$ can be obtained as
\begin{align}\label{short}
\mathrm{d} \chi_k(t)=\left(\frac{\sigma^2}{2} C_\epsilon -\kappa \chi_k(t-)\right) \mathrm{d} t+\sigma \sqrt{\chi_k(t-)} \mathrm{d} \tilde{L}_k^\alpha(t),
\end{align}
which resembles the Cox-Ingersoll-Ross (CIR) process widely used in finance to preserve positivity~\cite{r42}.

\subsection{The Composite Channel Model}
The proposed model captures a special non-Gaussian wireless communication scenario where the signal, after traveling a significant distance to the receiver's vicinity and experiencing path loss and long-term fading from reflections, undergoes scattering in the receiver's local area and is recombined, exhibiting short-term fading effects. The effects of long-term fading and short-term fading are connected in series, and the specific process is as follows.
\subsubsection{The Non-Gaussian long-term Fading Process}
The long-term fading component $\beta(t) \in \mathbb{R}$ is governed by the following SDE
\begin{equation} 
\mathrm{d}\beta(t) = -a (\beta(t) + b) \mathrm{d}t + \sigma_\beta \mathrm{d}L_\beta^\alpha(t), \, t\in\mathbb{R}^+
\end{equation}
as in \eqref{long}, where $a>0$, $b>0$, $\sigma_\beta>0$. The output signal power after the long-term fading stage is given by $e^{\beta(t)}p_{\text{in}}(t)$.

\subsubsection{The Non-Gaussian Short-term Fading Process}
For the short-term fading process $\chi(t)\in\mathbb{R}^+$, the evolution is described by
\begin{align}
\mathrm{d} \chi(t)=\left(\frac{\sigma_\chi^2}{2} C_\epsilon -\kappa \chi(t-)\right) \mathrm{d} t&+\sigma_\chi \sqrt{\chi(t-)} \mathrm{d} \tilde{L}_\chi^\alpha(t), \notag\\
&\quad \chi(0)=\chi_0, \quad t\in \mathbb{R}^+
\end{align}
as in \eqref{short}, with $\sigma_\chi>0$, $\kappa>0$. The $\tilde{L}_\chi^\alpha(t)$ is an independent S$\alpha$S Lévy process. To analyze the mean behavior of $\chi(t)$, taking expectations and noting the zero-mean property of the Lévy noise yields
\begin{align}
\mathbb{E}[\mathrm{d} \chi(t)]=\mathbb{E}\left[\left(\frac{\sigma_\chi^2}{2} C_\epsilon -\kappa \chi(t-)\right)\right] \mathrm{d} t.
\end{align}

By solving, it can be concluded that the expected value of $\chi(t)$ approaches a steady-state value $\mathbb{E}[\chi(t)]=\frac{\sigma_\chi^2}{2 \kappa} C_\epsilon$, as $t \rightarrow \infty$. To reflect the adaptation of $\chi(t)$ to its target mean value, the SDE is reformulated as
\begin{align}
\mathrm{d} \chi(t)=&\frac{\sigma_\chi^2}{2} C_\epsilon\left(1-\frac{\chi(t-)}{e^{\beta(t)} p_{\text {in}}(t)} \right) \mathrm{d} t\notag\\
&+\sigma_\chi \sqrt{\chi(t-)} \mathrm{d} \tilde{L}_\chi^\alpha(t), \quad \chi(0)=\chi_0, \quad t\in \mathbb{R}^+.
\end{align}

To avoid singularities when $p_{\text{in}}(t) = 0$, a small constant $0<\rho\ll1$ is introduced
\begin{align}
\mathrm{d} \chi(t)=&\frac{\sigma_\chi^2}{2} C_\epsilon\left(1-\frac{\chi(t-)}{e^{\beta(t)} p_{\text {in}}(t)+\rho} \right) \mathrm{d} t\notag\\
&+\sigma_\chi \sqrt{\chi(t-)} \mathrm{d} \tilde{L}_\chi^\alpha(t), \quad \chi(0)=\chi_0, \quad t\in \mathbb{R}^+.
\end{align}

\subsubsection{Slow-Fast Dynamics}
To capture the temporal disparity between long-term and short-term fading, a slow-fast dynamical system is introduced. In practical wireless scenarios, large-scale effects such as path loss and shadowing evolve slowly due to gradual changes in user mobility or terrain, whereas small-scale fading fluctuates rapidly due to multipath propagation and localized scattering. A small parameter \( 0 < \tau \ll 1 \) is introduced to characterize this time-scale separation. Specifically, \( \beta(t) \) evolves on the slow scale, while \( \chi(t) \) evolves on the fast scale. Following standard singular perturbation techniques, the SDE for \( \chi(t) \) is rescaled by \( \tau^{-1} \) in the drift term and \( \tau^{-1/2} \) in the diffusion term.

This framework enables clearer analysis of the interaction between fading components. The slow variation of \( \beta(t) \) allows it to be treated as quasi-static over short intervals, while \( \chi(t) \) responds rapidly to instantaneous variations in the effective gain \( e^{\beta(t)} p_{\text{in}}(t) \). Such a separation is critical for control design and performance evaluation where power control and adaptation operate at different time scales.

The resulting slow-fast system is expressed as, for $t \in \mathbb{R}^+$,
\begin{equation} \label{model long}
\mathrm{d}\beta(t) = -a (\beta(t) + b) \mathrm{d}t + \sigma_\beta \mathrm{d}L_\beta^\alpha(t), \quad \beta(0) = \beta_0,
\end{equation}
\begin{align}\label{model short}
\mathrm{d} \chi(t)=&\frac{1}{\tau}\frac{\sigma_\chi^2}{2} C_\epsilon\left(1-\frac{\chi(t-)}{e^{\beta(t)} p_{\text {in}}(t)+\rho} \right) \mathrm{d} t\notag\\
&+\frac{\sigma_\chi}{\sqrt{\tau}} \sqrt{\chi(t-)} \mathrm{d} \tilde{L}_\chi^\alpha(t), \quad \chi(0)=\chi_0.
\end{align}

The complete stochastic model for the concatenated channel, which captures both long-term and short-term fading effects with multiscale dynamics, i.e. $p_{\text{out}}(t)=\chi(t)$. This unified Lévy-driven state-space model captures both slow-varying large-scale fading and rapid small-scale fluctuations within a continuous-time stochastic framework. The model forms the foundation for subsequent power control and optimization strategies in interference-limited, non-Gaussian wireless environments.

\section{General fractional HJB control framework}\label{S4}
This section develops a general stochastic control framework for systems governed by non-Gaussian dynamics, particularly those driven by S$\alpha$S Lévy processes. To account for the discontinuities and heavy-tailed characteristics inherent in such processes, a fractional-order extension of the classical HJB equation is formulated. The resulting fractional HJB equation incorporates a nonlocal Riesz-type operator, derived from the infinitesimal generator of Lévy-driven SDEs, enabling robust optimal control design in interference-limited and non-stationary environments.

Furthermore, the existence and uniqueness of viscosity solutions to the fractional HJB equation are established using the comparison principle and the Perron’s method. The Riesz fractional operator captures the nonlocal jump behavior intrinsic to Lévy dynamics, thereby providing a rigorous theoretical foundation for constructing consistent and reliable control policies under heavy-tailed uncertainty.

\subsection{Derivation of the Fractional HJB Control Framework}

Consider a general system with state $x(t) \in \mathbb{R}^d$ described by the stochastic differential equation
\begin{align}\label{state}
\mathrm{d}x(s)&=b(x(s),u(s))\mathrm{d}s+\sigma(x(s))\mathrm{d}L^\alpha(s),\notag\\
x(t)&=x_0,
\end{align}
where $x_0$ is the given initial state at time $t>0$, and $T$ is the 
terminal time. $u(\cdot) \in \mathcal{U}$ is an admissible control, adapted to the natural filtration and satisfying measurability and integrability conditions. $b(x, u)$ represents the drift term, which is Lipschitz in $x$ and measurable in $u$. $\sigma(x)$ denotes the state-dependent noise intensity, which is Lipschitz with uniformly bounded determinant $|\det\sigma(x)|>0$. $L^\alpha$ is a vector of S$\alpha$S Lévy processes with Lévy measure $\nu_\alpha(\mathrm{d}\xi)=C_\alpha|\xi|^{-d-\alpha}\mathrm{d}\xi$. Among them, $C_\alpha=\frac{\alpha 2^{\alpha-1} \Gamma\left(\frac{\alpha+d}{2}\right)}{\pi^{d / 2} \Gamma\left(1-\frac{\alpha}{2}\right)}$ represents the normalization factor.

The objective is to minimize the cost functional
\begin{equation}
J(u; t) = \mathbb{E} \left[ \int_t^T L(s, x(s), u(s)) \mathrm{d}s + g(x(T)) \right],
\end{equation}
where $L$ is a general running cost and $g$ is a terminal cost. 

To utilize the principles of DP, at any intermediat time $t\in[0,T]$, the value function is defined by
\begin{equation}
    V(t,x)=\min _{u(\cdot)\in\mathcal{U}} J(u;t). 
\end{equation}
For a small time increment $\Delta t$, 
\begin{align}\label{DP}
V\left(t, x\right)=&\min _{u(\cdot)} \mathbb{E}\left[\int_t^{t+\Delta t} L\left(s,x(s),u(s)\right) \mathrm{d} s\right.\notag\\
&\left.+V\left(t+\Delta t, x(t+\Delta t) \right) \mid x(t)=x\right].
\end{align}

In the channel model, the use of S$\alpha$S Lévy processes with $\alpha < 2$ leads to infinite variance, rendering classical Itô calculus and second-order expansions inapplicable. However, for symmetric, isotropic Lévy processes with zero drift and no Gaussian component, the infinitesimal generator can still be derived via the Lévy–Khintchine formula. In the homogeneous case with constant noise intensity, this generator reduces to the fractional Laplacian \( (-\Delta)^{\alpha/2} \), also known as the Riesz fractional derivative.

In contrast, \eqref{state} involves a state-dependent noise coefficient \( \sigma(x) \), resulting in a non-homogeneous process. In this setting, the standard Riesz operator with kernel \( \|x - y\|^{-d - \alpha} \) is no longer valid. To accommodate spatial inhomogeneity, a generalized Riesz fractional operator must be introduced, which can be defined in two equivalent forms. From the kernel perspective, it acts on $f\in C^2(\mathbb{R}^d)$, both the function itself and its first and second partial derivatives are globally bounded.
\begin{equation}
\nabla_\sigma^\alpha f(x)=\mathrm{P}. \mathrm{V}. \int_{\mathbb{R}^d}[f(x)-f(y)] K_\sigma(x, y) \mathrm{d} y,
\end{equation}
where $\mathrm{P.V.}$ represents the principal integral. The kernel $K_\sigma(x, y)=\frac{C_\alpha}{\left\|\sigma(x)^{-1}(x-y)\right\|^{d+\alpha}} \frac{1}{|\operatorname{det} \sigma(x)|}$ explicitly depends on $\sigma(x)$. From the Lévy-Khintchine standpoint, let
\begin{align}
\mathcal{A}_\sigma f(x) =& \int_{\mathbb{R}^{d} \setminus \{0\}} \left[ f(x + \sigma(x)\xi) - f(x) \right.\notag\\
    & \left. - \mathbf{1}_{|\xi| < 1} \sigma(x)\xi \cdot \nabla f(x)  \right] \nu(\mathrm{d}\xi),
\end{align}
with $\nu(\mathrm{d} \xi)=\frac{C_\alpha}{|\xi|^{d+\alpha}} \mathrm{d} \xi$ the Lévy measure of $L^\alpha$. The indicator function $\mathbf{1}_{|\xi| < 1}$ adjusts the expansion for small jumps to ensure convergence. 
The assertion that $\nabla_\sigma^\alpha$ and $\mathcal{A}_\sigma$ are equivalent descriptions is based on a formal change of variables, $y=x+\sigma(x) \xi$, within the integral definition of the generator. This equivalence holds rigorously for functions $f \in C_b^2\left(\mathbb{R}^d\right)$ (i.e., twice continuously differentiable functions with bounded derivatives) under the following two key assumptions on the noise matrix $\sigma(x)$ (see e.g., \cite{r43,matter}).

\emph{Uniform Ellipticity}:  There exist constants $0<\lambda \leq \Lambda<\infty$ such that the matrix $\sigma(x) \sigma(x)^T$ satisfies $\lambda I \leq \sigma(x) \sigma(x)^T \leq \Lambda I$ for all $x \in \mathbb{R}^d$. The lower bound ensures that $\sigma(x)$ is uniformly invertible, keeping its determinant bounded away from zero and its inverse well-behaved, which is essential for the kernel $K_\sigma(x, y)$ to be well-defined. The upper bound controls the growth of the noise intensity.

\emph{Regularity}: The matrix-valued function $\sigma(x)$ is Lipschitz continuous, as is standard for the well-posedness of the SDE in \eqref{state}. This smoothness condition is necessary for the validity of the change of variables transformation.

Under the above two conditions, which are standard in the theory of integro-differential equations, the kernel-based representation (36) and the Lévy-Khintchine generator form (37) are mathematically equivalent.

Assuming \( V \) is sufficiently smooth and leveraging the independent increment property of Lévy processes, a nonlocal Taylor expansion is applied to \( V(t + \Delta t, x(t + \Delta t)) \). This expansion accounts for jump effects via an integral term weighted by the Lévy measure. Applying Itô’s formula for Lévy processes, the increment of \( V \) can be expressed as
\begin{align}\label{Ex}
V(t+\Delta t, x(t+\Delta t))=&V(t, x)+\partial_t V(t, x) \Delta t\notag\\
&+\nabla_x V(t, x) \cdot b(x, u) \Delta t\notag\\
&+\nabla^\alpha_\sigma V(t, x) \Delta t+o(\Delta t),
\end{align}
where the $\nabla^\alpha_\sigma V(t, x)$ arises from the nonlocal integral component associated with the jumps of the Lévy process. From the generator perspective, this term represents the action of the infinitesimal operator.

It is important to note that the expansion in \eqref{Ex} is a heuristic step that assumes sufficient smoothness of the value function $V$ to motivate the form of the resulting HJB equation. We recognize that $V$ is generally not smooth for jump-diffusion processes. Therefore, once the equation is derived, we will abandon this smoothness assumption and seek a solution in the rigorous weak sense of viscosity solutions, which is specifically designed to handle such non-smooth cases.

Substituting \eqref{Ex} into \eqref{DP}, subtracting $V$ from both sides, and taking the limit as $\Delta t\rightarrow 0$, we arrive at
\begin{align}
    0 =& \min_{u(\cdot)} \left\{L(x, u)+\partial_t V(t, x)\right.\notag\\
    &\Big.+\nabla_x V(t, x) \cdot b(x, u)+\nabla^\alpha_\sigma V(t, x)  \Big\} .
\end{align}
When $\Delta t \to 0$, the term $\frac{o(\Delta t)}{\Delta t}\to 0$. Introduce the Legendre transform and define the Hamiltonian quantity as follows
\begin{equation}
H(x, p):=\min _{u(\cdot)}p \cdot b(x, u)+L(x, u),
\end{equation}
where $p=\nabla_x V(t,x)$. Overall, the fractional HJB equation yields 
\begin{equation}\label{eq1}
    \partial_t V(t,x) + H(\nabla_x V(t,x)) +\nabla^\alpha_\sigma V(t,x) = 0,
\end{equation}
with the terminal condition $g(x(T))$.

\subsection{Existence and Uniqueness of Viscosity Solutions}
Firstly, the concepts of different types of  viscosity solutions applicable to fractional operators are introduced. For a function space, the notations $\USC$ and $\LSC$ denote upper and lower semicontinuous functions, respectively.

\begin{definition}[Viscosity Subsolution]\label{def1} 
A function $V \in \USC([0, T]\times \mathbb{R}^{d})$ is a viscosity subsolution of \eqref{eq1} if, for any $(t_0,x_0) \in (0, T] \times \mathbb{R}^d$ and any test function $\phi \in C^{1,2}((0, T] \times \mathbb{R}^{d})$ such that $V - \phi$ attains a local maximum at $(t_0,x_0)$, the following inequality holds 
\begin{align}\label{eq2} 
\partial_t \phi + H\left(\nabla_{x} \phi\right) + \nabla^\alpha_\sigma[\phi] \leq 0.
\end{align}
\end{definition}

\begin{definition}[Viscosity Supersolution]\label{def2} 
A function $V \in \LSC([0, T] \times \mathbb{R}^d)$ is a viscosity supersolution of \eqref{eq1} if, for any $(t_0,x_0) \in (0, T] \times \mathbb{R}^d$ and any test function $\phi \in C^{1,2}((0, T] \times \mathbb{R}^d)$ such that $V - \phi$ attains a local minimum at $(t_0,x_0)$, the following inequality holds
\begin{align}
\partial_t \phi + H\left(\nabla_{x} \phi\right) + \nabla^\alpha_\sigma[\phi] \geq 0.
\end{align} 
\end{definition}

\begin{definition}[Viscosity Solution]\label{def3} 
A function $V \in C([0, T] \times \mathbb{R}^d)$ is called a viscosity solution of \eqref{eq1} if it is both a viscosity subsolution and a viscosity supersolution. \end{definition}

To ensure the adequacy of the problem, the following assumptions are invoked.

\textbf{(A1)} The Hamiltonian $H: \mathbb{R}^d \rightarrow \mathbb{R}$ is continuous on $\mathbb{R}^d$.

\textbf{(A2)} There exists a constant $L_H > 0$ such that for all $x$, and $q_1, q_2\in \mathbb{R}^d$
\begin{equation*}
\left| H(x, q_1) - H(x, q_2) \right| \leq L_H || q_1 - q_2 ||.
\end{equation*}

\textbf{(A3)} There exists a constant $C_H \geq 0$ such that for all $x$, and $q\in\mathbb{R}^d$
\begin{equation*}
\left| H(x, q) \right| \leq C_H (1 + || q ||).
\end{equation*}

\textbf{(A4)} The terminal condition $V(T, x)=g(x(T))$ is bounded and uniformly continuous, that is, $g(x(T))\in \BUC(\mathbb{R}^d)$.

\textbf{(A5)} For each test function $\phi\in C^{1,2}$ the following kernel-growth condition is satisfied (see, e.g., \cite{Barles}):
\begin{equation*}
    \int_{\mathbb{R}^d\setminus\{0\}}|\phi(x+\sigma(x)\xi)-\phi(x)|\,\nu(\mathrm{d}\xi)<\infty.
\end{equation*}

\subsubsection{Existence of Viscosity Solutions}
In order to establish the existence of viscosity solutions for \eqref{eq1}, the Perron’s method suitable for the non-local fractional operator is adopted. The key idea is to construct an upper envelope of subsolutions and show that it is itself a solution, relying on the comparison principle. Specifically, by using Riesz fractional derivatives as zero or bounded constants and linear functions, the existence of global subsolutions and supersolutions has been proven. The adaptation of the Perron's method to the non-local case is crucial due to the nature of the fractional derivative, which introduces memory effects and non-local interactions that are not present in local differential operators. 

\begin{lemma}\label{lem1} Let $W_1=\sup _{x \in \mathbb{R}^d}\left|g(x(T))\right|$ and define  
\begin{equation} \underline{V}(t, x) = - W_1 - C (T - t), \, \overline{V}(t, x) = W_1 + C (T - t), \end{equation}
where $C \geq 2C_H$ and $C_H$ is the linear growth constant of the Hamiltonian $H$. Then, $\underline{V}$ is a viscosity subsolution and $\overline{V}$ is a viscosity supersolution of \eqref{eq1} on $[0, T] \times \mathbb{R}^d$. 
\end{lemma}

\begin{proof} 
For $\underline{V}$, one has  $\partial_t \underline{V} = C$, $\nabla_{x} \underline{V} = 0$, and $\nabla^\alpha_\sigma[\underline{V}] = 0$. Substituting into the fractional HJB equation yields
\begin{equation}
\partial_t \underline{V} + H(\nabla_x \underline{V}) + \nabla_\sigma^\alpha [\underline{V}] = C + H(x, 0).
\end{equation}
Using assumption (A3), it follows that
\begin{equation}
    C + H(x, 0) \leq C + C_H\leq 2C_H \leq C.
\end{equation}
Thus
\begin{equation}
\partial_t \underline{V} + H(x, \nabla_x \underline{V}) + \nabla_\sigma^\alpha [\underline{V}] \leq 0,
\end{equation}
which confirms that $\underline{V}$ is a viscosity subsolution. The proof for $\overline{V}$ follows analogously by verifying the inequality for a supersolution.
\end{proof}
 
\begin{theorem}[Existence of Viscosity Solution]\label{thm3} 
Under assumptions \textbf{(A1)}-\textbf{(A5)}, the viscosity solution $V \in \BUC([0, T] \times \mathbb{R}^d)$ of the fractional HJB equation in \eqref{eq1} exists with terminal condition $g(x(T))$. 
\end{theorem}
\begin{proof}
    The existence of a viscosity solution is established using Perron's method, which is a standard technique for constructing solutions to non-local HJB equations. The proof follows the methodology detailed in the foundational work by \cite{Barles}, which adapts the method to second-order integro-differential equations. The key steps, as adapted to our specific Hamiltonian, are outlined in Appendix A.
\end{proof}

\subsubsection{Uniqueness of Viscosity Solutions}
Uniqueness of the viscosity solution to the fractional HJB equation in \eqref{eq1} is established by proving a comparison principle, ensuring that no two distinct solutions can satisfy the same terminal condition. 

\begin{theorem}[Comparison Principle]\label{thm1} 
Suppose $V_1$ is a viscosity subsolution and $V_2$ is a viscosity supersolution of \eqref{eq1}. If they satisfy the terminal condition $V_1(T,x)\leq V_2(T,x)$ for all $x\in\mathbb{R}^d$, then $V_1\leq V_2$ on $[0, T] \times \mathbb{R}^d$. 
\end{theorem}

\begin{proof}
    The proof relies on establishing a comparison principle, which is the cornerstone of uniqueness for viscosity solutions. The argument is a non-trivial adaptation of the classic "doubling the variables" method, extended to handle the non-local fractional operator. The detailed procedure for such operators is established in the literature \cite{Barles} and \cite{Lions}. We summarize the main arguments of this adaptation in Appendix B.
\end{proof}

\begin{corollary}[Uniqueness of Viscosity Solution]\label{cor1} 
Under assumptions \textbf{(A1)}-\textbf{(A5)}, the viscosity solution to the fractional HJB equation in \eqref{eq1} is unique, with terminal condition $g(x(T))$ on $[0, T] \times \mathbb{R}^d$. 
\end{corollary}

\begin{theorem}[Existence and Uniqueness of Viscosity Solution]\label{thm2} 
Through Theorems \ref{thm3} and \ref{thm1}, the viscosity solution $V \in \BUC([0, T] \times \mathbb{R}^d)$ of the fractional HJB equation in \eqref{eq1} exists and is unique. 
\end{theorem}

By adapting the Perron's method and the comparison principle to the non-local fractional setting, the existence and uniqueness of viscosity solutions to the fractional HJB equation arising from optimal control problems with Lévy flights are rigorously established. These mathematical results have important implications for the framework's viability as a control design tool. The existence of a solution confirms that the optimal control problem is well-posed, ensuring that a valid value function can be found for all possible channel states. The uniqueness of this solution is also critical, as it ensures that the derived optimal control policy is unambiguous and consistent, which is a necessary condition for the stable and predictable operation of the controller.

\section{Numerical Simulations}\label{S5}
This section presents numerical simulations to validate the effectiveness and robustness of the proposed fractional HJB framework for optimizing base station transmit power control in multi-cell multi-user downlink systems operating simultaneously on the same frequency. The framework seeks a globally optimal policy by solving a centralized control problem over the joint state space of the network. Specifically, considering that each BS $i$ broadcasts to multiple UEs $\ell$, the instantaneous power received by users experiences long-term fading and short-term fading.

The primary objective of the simulations is to assess the ability of DP and value iteration techniques, applied within the fractional HJB framework, to determine optimal transmission power control strategies that balance the trade-off between communication quality and power consumption. The simulations will explore how the system adapts to changing channel conditions and evaluate the performance of the power control strategies derived from the proposed model. Through iterative optimization, we seek to demonstrate the utility of the fractional HJB equation in managing the interference and fading dynamics in a non-Gaussian communication environment.

The long-term fading from BS $i$ to UE $\ell$ is given by
\begin{equation} 
\mathrm{d}\beta_{i\ell} (t) = -a (\beta_{i\ell}(t) + b) \mathrm{d}t + \sigma_{\beta} \mathrm{d}L_{\beta}^\alpha(t), 
\end{equation} 
with $\beta_{i\ell}(0) = (\beta_{i\ell})_0$. 

The short-term fading process, representing the received power $p_{i\ell}(t)$ at UE $\ell$ from BS $i$, follows
\begin{align}
\mathrm{d} p_{i\ell}(t)=&\frac{1}{\tau}\frac{\sigma_p^2}{2} C_\epsilon\left(1-\frac{p_{i\ell}(t-)}{e^{\beta_{i\ell}(t)} p_{\text{in},i}(t)+\rho} \right) \mathrm{d} t\notag\\
&+\frac{\sigma_p}{\sqrt{\tau}} \sqrt{p_{i\ell}(t-)} \mathrm{d} \tilde{L}_{p}^\alpha(t), \quad p_{i\ell}(0)=\left(p_{i\ell}\right)_0,
\end{align}
where the parameters are consistent with \eqref{model short}. $p_{\text{in},i}(t)=p_{0,i}+u_i(t)$ denotes the transmit power after applying control $u_i(t)$. The term $e^{\beta_{i\ell}(t)} p_{\text{in},i}(t)$ represents the effective power after experiencing long-term fading.

Therefore, the state vector in \eqref{state} can be specifically written as
$x(t)=\left(\begin{array}{lll}\beta_{i\ell}(t) & p_{i \ell}(t)\end{array}\right)^{\top}$ with corresponding SDE
\begin{align}\label{sstate}
\mathrm{d} x(t)=&\binom{-a\left(\beta_{i \ell}(t)+b\right)}{\frac{1}{\tau} \frac{\sigma_{p}^2}{2} C_{\epsilon}(1-\frac{p_{i\ell}(t-)}{e^{\beta_{i \ell}(t)} p_{\text{in},t}(t)+\rho})} \mathrm{d} t \notag\\
&+\left(\begin{array}{cc}
\sigma_{\beta} & 0 \\
0 & \frac{\sigma_p}{\sqrt{\tau}} \sqrt{p_{i\ell}(t-)}
\end{array}\right) \mathrm{d} L^\alpha(t),
\end{align}
with $L^\alpha(t)=(\begin{array}{ll}L_{\beta}^\alpha(t) & \tilde{L}_p^\alpha(t)\end{array})^{\top}$.

The signal quality from BS $i$ to UE $\ell$ is characterized by the signal to interference plus noise ratio (SINR)
\begin{equation}
\gamma_{i\ell}(t)=\frac{p_{i\ell}(t)}{\sum_{k=1, k \neq i}^N  p_{k\ell}(t)+\eta},
\end{equation}
where all variables are defined as previously, and $\eta$ is the additive thermal noise.

To formalize the control objective, a stage-wise cost function is defined to balance communication performance and power efficiency. The formulation incorporates three key components:  (i) a rate utility term that rewards higher SINR to improve throughput;  (ii) a penalty term that discourages quality-of-service (QoS) violations when SINR drops below a predefined threshold; and  (iii) a power cost term that penalizes excessive transmission power. The resulting cost function is expressed as
\begin{align}
    &L\left(t, u_i, u_{-i}\right)\notag\\
    &=\sum_{(i,\ell)}\left(-\log_2\left(1+\gamma_{i\ell}(t)\right)+\varsigma\left(r_{\text{th}}-\gamma_{i\ell}(t)\right)^++\lambda p_{\text{in},i}(t)\right),
\end{align}
where $u_i$ denotes the control of BS $i, u_{-i}$ denotes the controls of all other BSs excluding BS $i$, $\varsigma>0$ is the punishment intensity coefficient, $(x)^+=\max(x,0)$ means that if $x>0$, it is $x$, otherwise 0, and $\lambda>0$ represents the cost of controlling the transmission power consumption.

To map this specific problem to the general framework in Section \ref{S4}, we explicitly define the system components. The full state vector $x(t)$ is the concatenation of all ($\beta_{i \ell}(t), p_{i \ell}(t)$) pairs, and the control vector $u(t)$ is the collection of all BS controls $u_i(t)$. The system drift $b(x, u)$ and diffusion $\sigma(x)$ are then block-structured, composed of the terms in  \eqref{sstate} for each independent link. The co-state vector $p=\nabla_x V(t, x)$ is partitioned consistently with $x$, where the components $p_{\beta_{i \ell}}$ and $p_{p_{i \ell}}$ are the partial derivatives of the value function with respect to the states $\beta_{i \ell}$ and $p_{i \ell}$, respectively. The Hamiltonian for the centralized system is given by
\begin{align}
		&H(x, p)\notag\\
        &=\min _{u \in \mathcal{U}}\left\{\sum_{i=1}^N \sum_{\ell=1}^M\left[p_{\beta_{i \ell}} b_{\beta_{i \ell}}(x)+p_{p_{i \ell}} b_{p_{i \ell}}\left(x, u_i\right)\right]+L(x, u)\right\},
\end{align}
where $b_{\beta_{i \ell}}(x)=-a\left(\beta_{i \ell}+b\right)$, 

\noindent $b_{p_{i \ell}}\left(x, u_i\right)=\frac{1}{\tau}\frac{\sigma_p^2}{2} C_\epsilon\left(1-\frac{p_{i\ell}}{e^{\beta_{i\ell}} (p_{0,i}+u_i)+\rho}\right)$.

This instantiated Hamiltonian is the direct input to the fractional HJB equation \eqref{eq1} that we solve numerically.

\subsection{Simulation Parameter Settings}
\begin{table}[t]
\centering
\caption{Explanation of Numerical Simulation Parameters}
\label{Table1}
\begin{tabular}{cccccc}
\hline
\textbf{Parameter} & \textbf{Value} & \textbf{Parameter} & \textbf{Value} & \textbf{Parameter} & \textbf{Value} \\
\hline
$N$ & 3 & $M$ & 3 & $u_{\text{step}}$ & 1 \\
$p_{\text{max}}$ & 15 & $p_{\text{min}}$ & 0  & $\eta$ & 1\\
$N_t$ & 100 & $dt$ & 0.1 & $T$ & 10 \\
$r_{\text{th}}$ & 1.5 & $\varsigma$ & 1 & $\lambda$ & 0.1 \\
$a$ & 0.1 & $e^{b_{i\ell}},i=\ell$ & 0.5 & $e^{b_{i\ell}},i\neq\ell$ & 0.1\\
$\alpha$ & 1.8 & $\sigma_{\beta}$ & 0.1 & $\sigma_{p}$ & 0.3 \\
$C_{\epsilon}$ & 1 & $\tau$ & 0.01 & $\rho$ & $10^{-3}$ \\
\hline
\end{tabular}
\end{table}

In all numerical experiments, each BS’s transmit power is constrained to the interval [0,15] W, which reflects typical pico-/microcell deployment scenarios as recommended in 3GPP TS 36.101 \cite{3GPP} and Goldsmith \cite{Goldsmith}. The median large-scale fading gains are determined by the mean-reversion parameter $b_{i \ell}$ in the long-term fading SDE, which sets the long-term average of the process. For the desired (serving) link, the median gain is set to 0.5 ($\approx$ –3 dB), while for interfering links it is set to 0.1 ($\approx$ –10 dB), consistent with urban microcell measurements reported in Rappaport \cite{Rappaport}. The SINR threshold is set to $r_{\mathrm{th}}=1.5$ ($\approx$ 1.76 dB) for each UE. This value is selected as it represents a standard baseline for ensuring Quality of Service (QoS); it is the minimum linear-scale SINR required to support robust modulation schemes such as QPSK with a Block Error Rate (BLER) of 10\% or less, as detailed in Goldsmith \cite{Goldsmith} and TS 36.101 \cite{3GPP}.

To obtain the optimal control strategy, value iteration is conducted over 100 rounds using an explicit action space. This action space is constructed from a set of discrete power adjustments, where the control input $u_i(t)$ is selected from a discrete set of adjustments within a range of $\pm u_{\text {step}}$, where $u_{\text{step}}=1 \mathrm{W}$. Other parameters related to the SDEs (e.g., $\alpha$, $\sigma_{\beta}$, $\sigma_{\text{out}}$, $\tau$) are selected within numerically stable ranges to reflect heavy-tailed noise characteristics while ensuring that the system remains physically realistic and computationally tractable.
Table \ref{Table1} summarizes the key parameters used in the simulation.

\subsection{Simulation Results and Discussions}
\begin{figure*}[htbp]
\centering
\subfloat[Initial]{\includegraphics[width=2.3in]{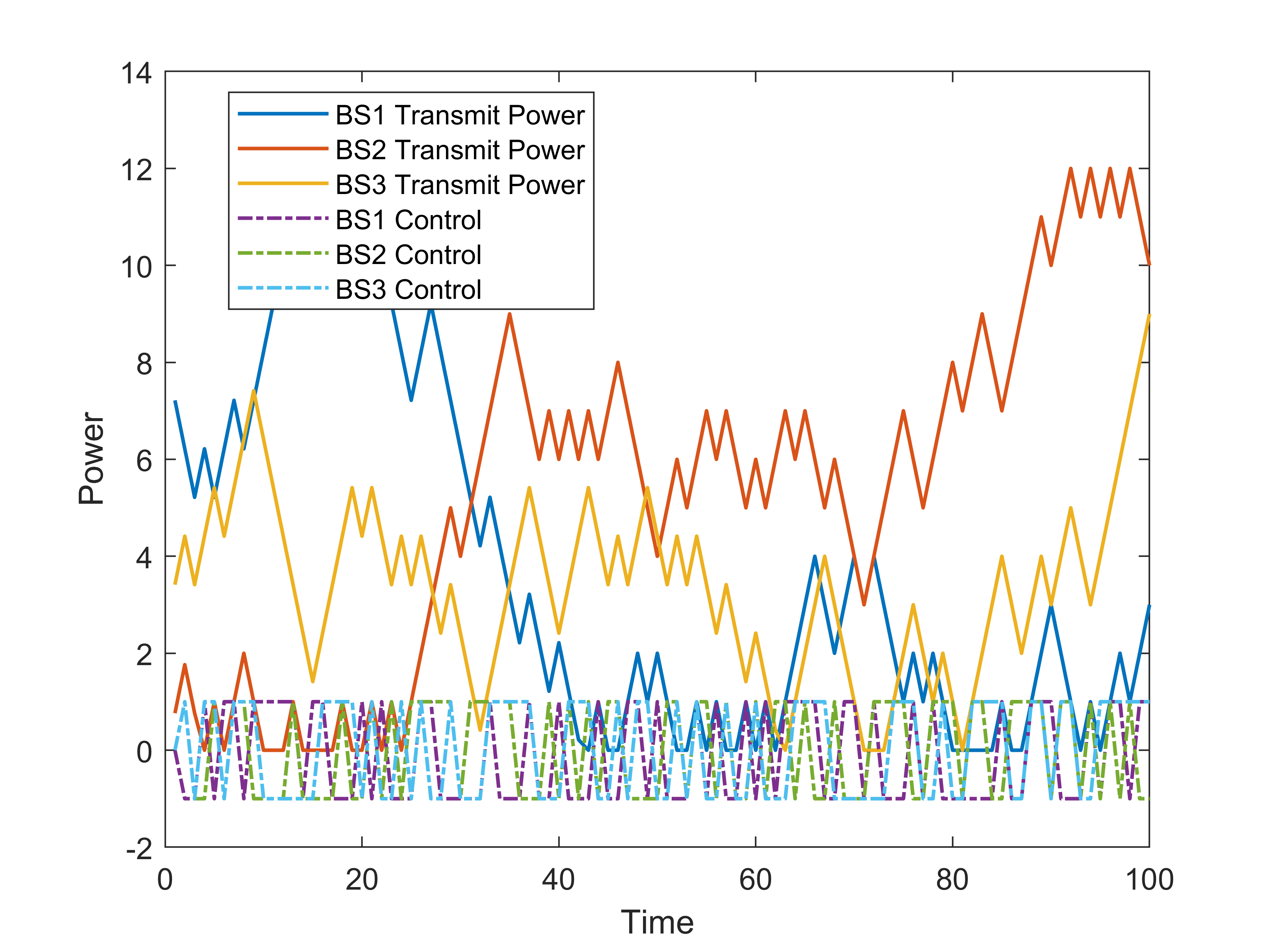}%
\label{fig1}}
\hfil
\subfloat[50 rounds]{\includegraphics[width=2.3in]{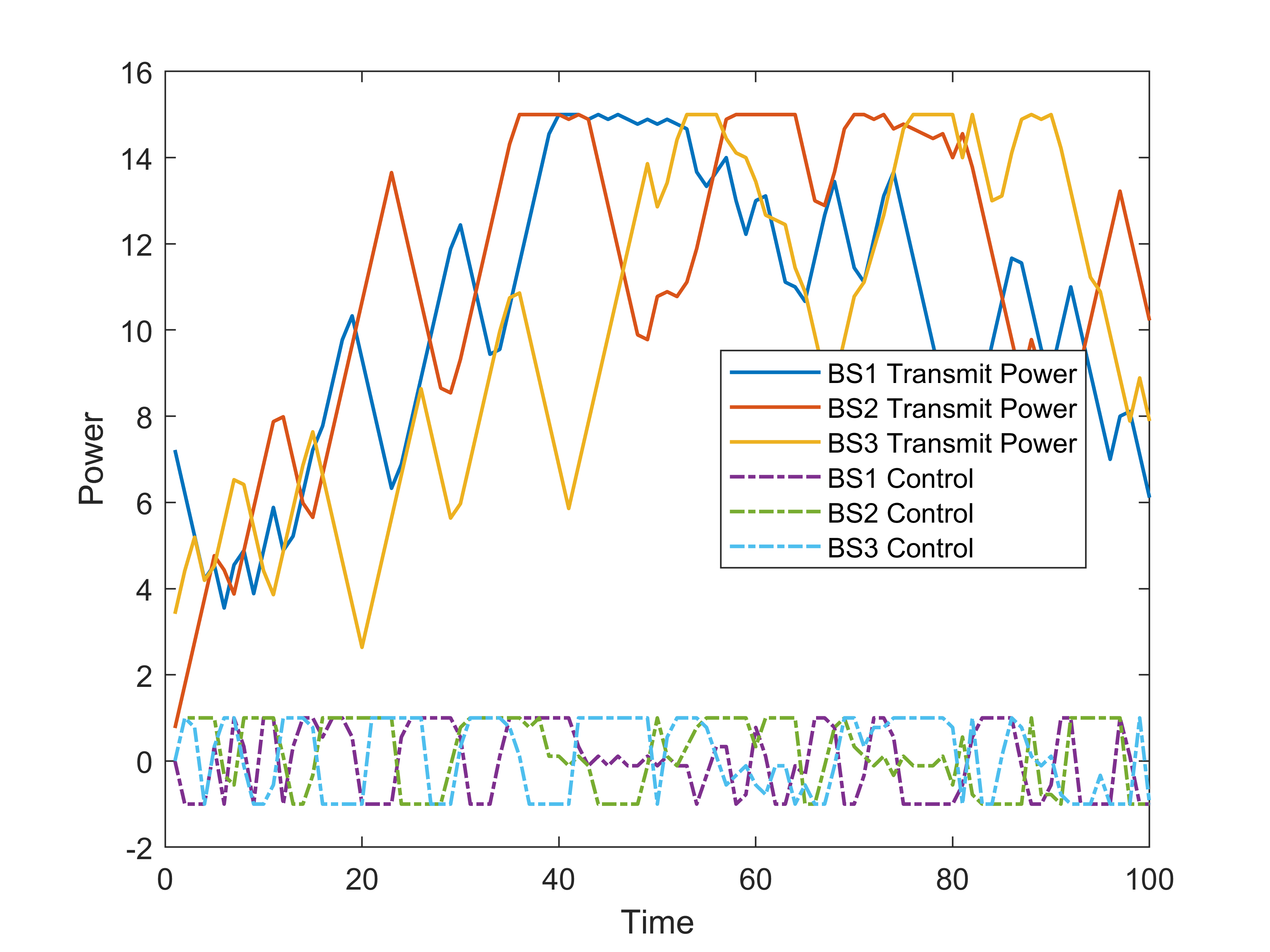}%
\label{fig2}}
\hfil
\subfloat[100 rounds]{\includegraphics[width=2.3in]{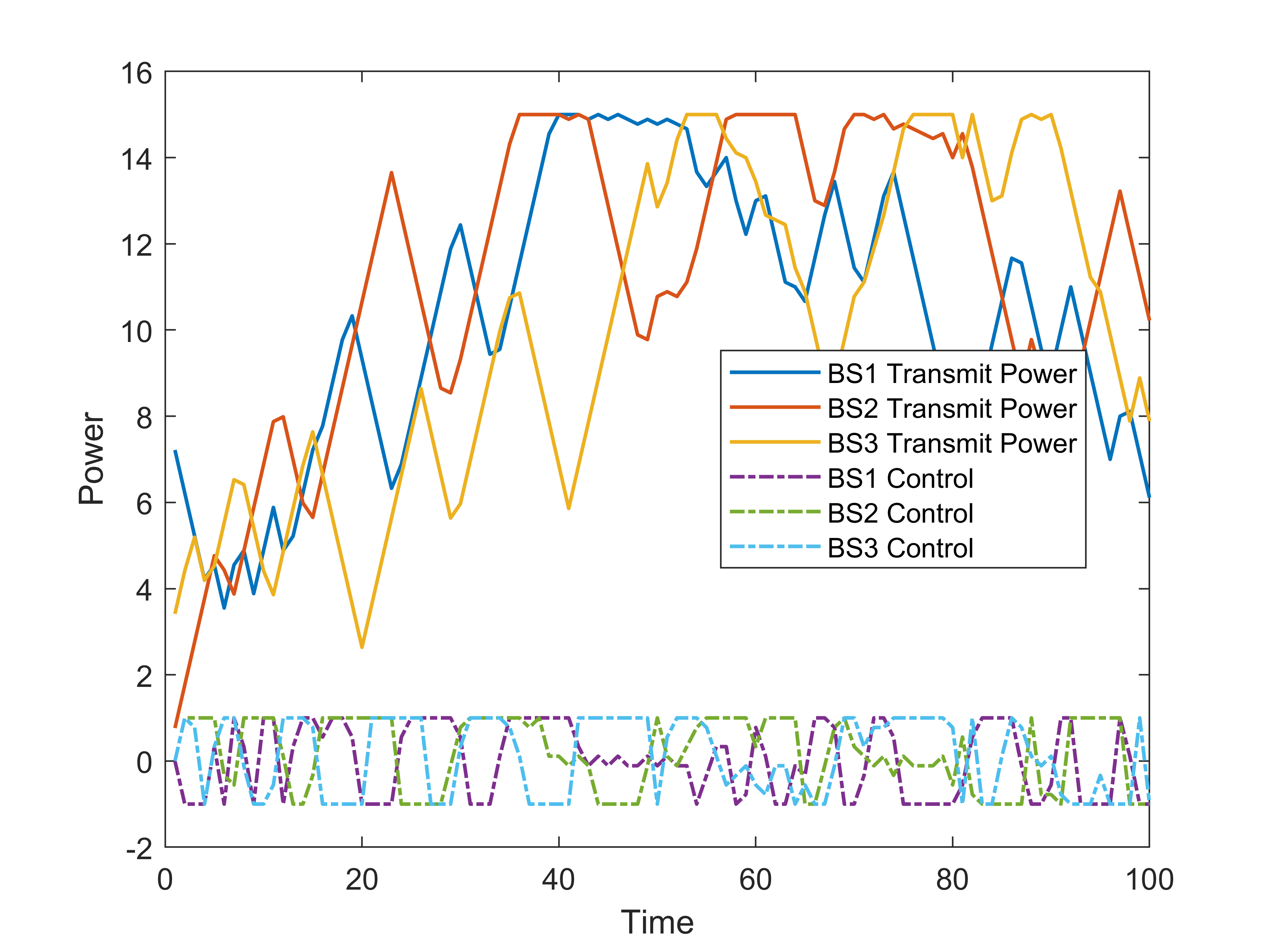}%
\label{fig3}}
\caption{Evolution of the BS Transmit Power and Control Actions after DP Iterations with $\alpha=1.8$.}
\label{Fig1}
\end{figure*}

\begin{figure*}[htbp]
\centering
\subfloat[Value]{\includegraphics[width=2.3in]{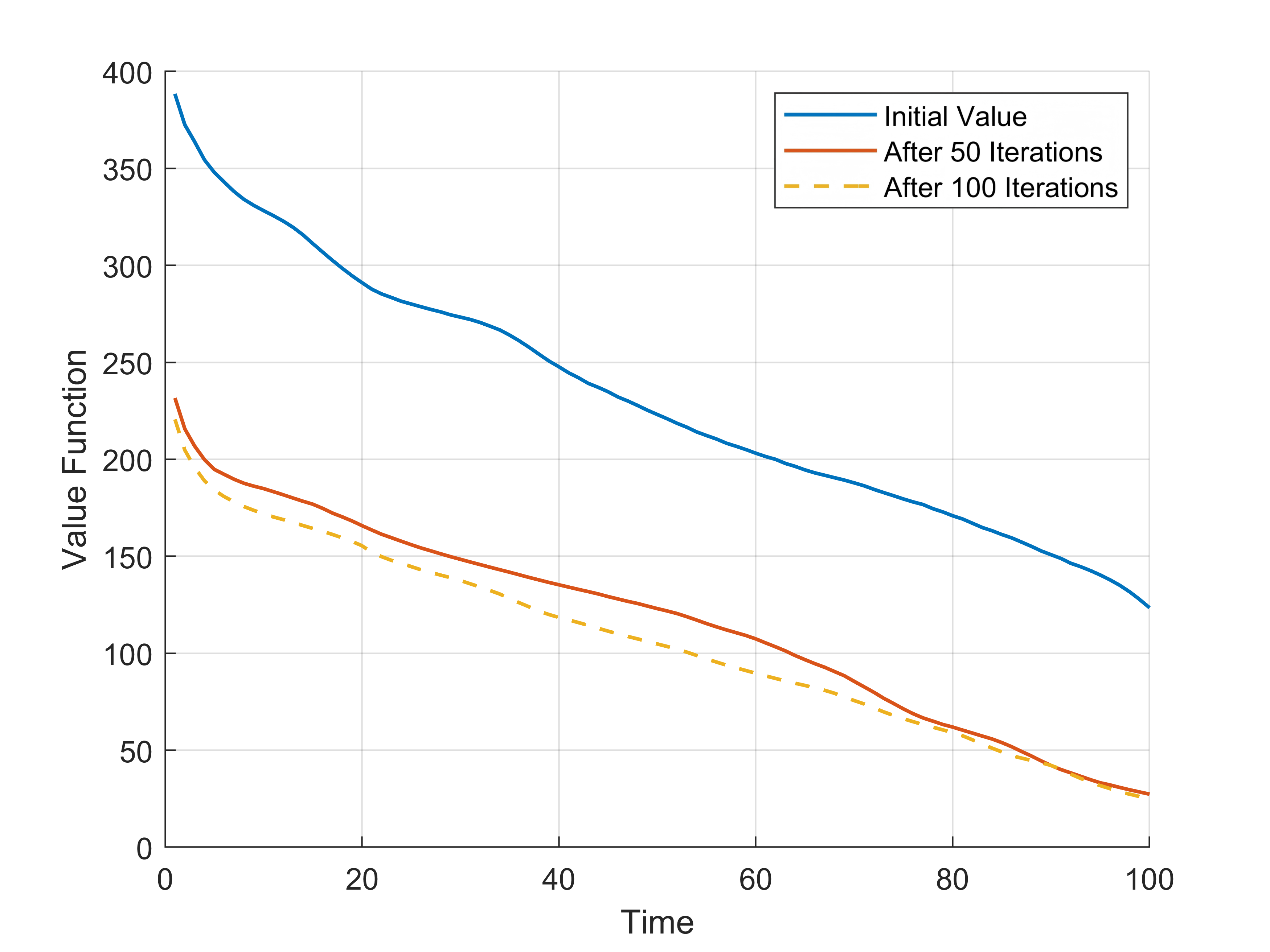}%
\label{fig4}}
\hfil
\subfloat[SINR: 50 rounds]{\includegraphics[width=2.3in]{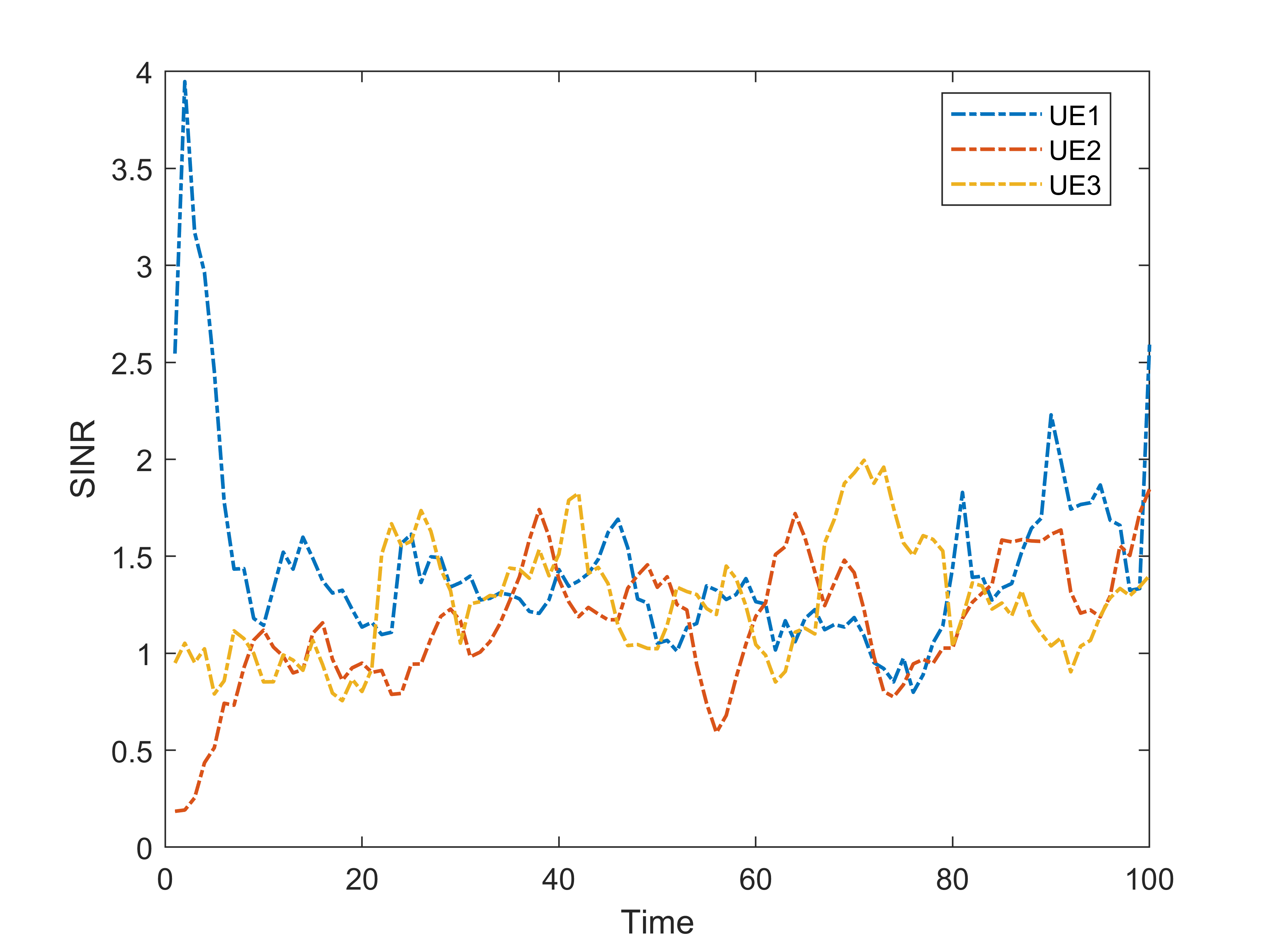}%
\label{fig5}}
\hfil
\subfloat[SINR: 100 rounds]{\includegraphics[width=2.3in]{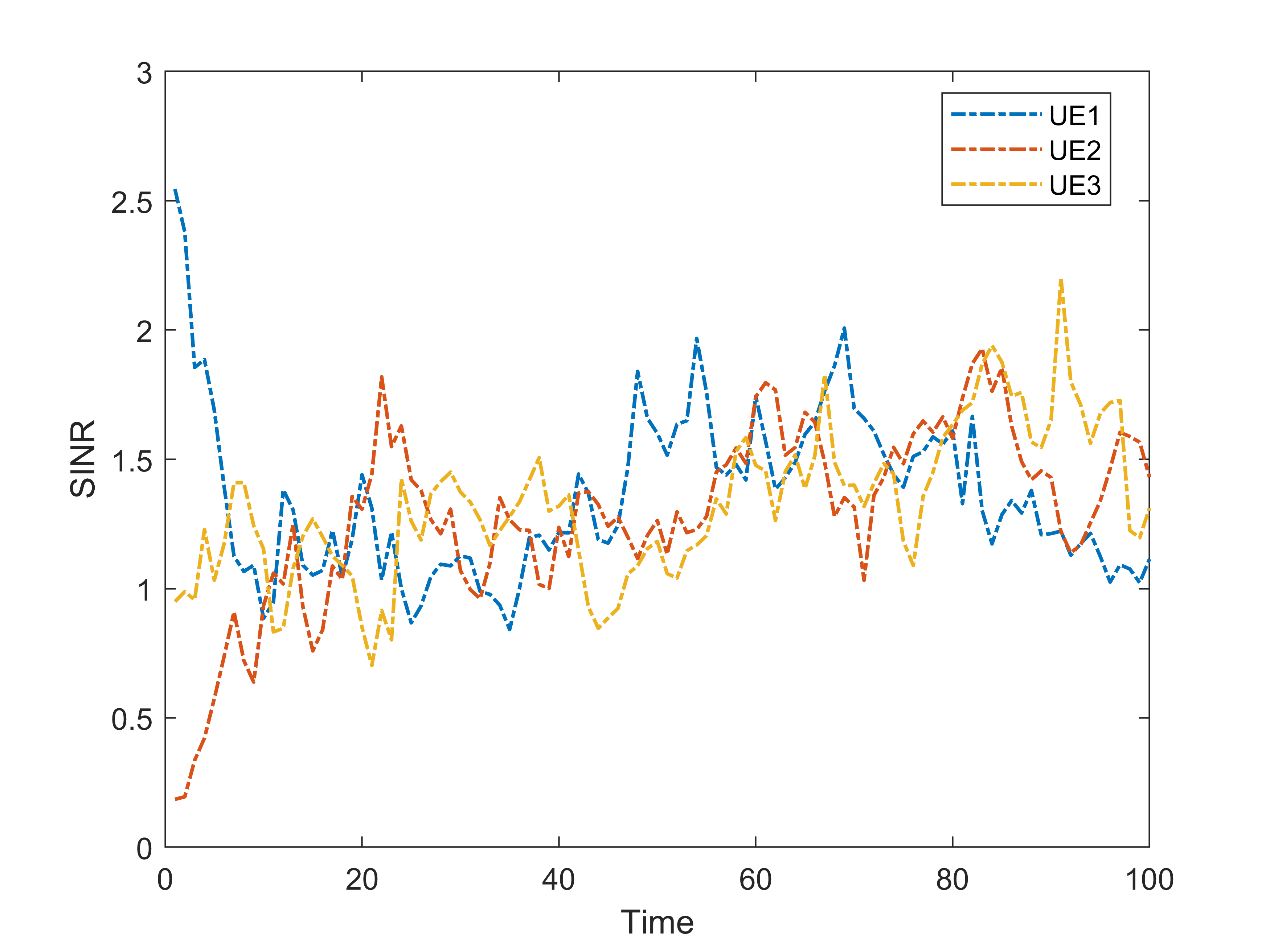}%
\label{fig6}}
\caption{Evolution of the Value Function and SINR Trajectories after DP Iterations with $\alpha=1.8$.}
\label{Fig2}
\end{figure*}

Fig.~\ref{Fig1} illustrates the evolution of each base station’s transmission power and control actions at the initial state, after 50 iterations, and after 100 iterations. Initially, actions are nearly random, leading to fluctuating and uncoordinated power levels, resulting in inefficient energy usage and increased interference.

After 50 iterations, as shown in Fig.~\ref{fig2}, each base station begins adjusting its power more proactively in response to both local fading and interference from others. The power trajectories exhibit fewer abrupt swings, and control actions become more structured. Notably, the optimal policy exhibits a highly coordinated power adaptation pattern: when one base station increases its power, others often reduce theirs. This behavior is a direct result of the centralized optimization, where the dynamic programming algorithm finds a globally optimal strategy to minimize the total system cost. The controller learns to avoid simultaneous high-power transmissions that would degrade the overall network SINR through mutual interference.

After 100 iterations, the control strategy converges as shown in Fig.~\ref{fig3}: power profiles stabilize into bounded ranges with coordinated peak–valley alternation. This indicates that the converged control policy now effectively accounts for the multistep channel dynamics and interference feedback for each base station. 
Control actions are now consistently applied to maintain SINR above threshold while minimizing power cost. This power balancing behavior reflects the globally optimal solution, where the centralized policy adapts each base station's transmission strategy based on the joint state of all channel conditions and the interference they induce on one another.

Fig. \ref{fig4} illustrates that the value function $V(t)$ gradually decreases as the number of iterations increases. As the cost function is composed entirely of non-negative components—including SINR outage penalties, power consumption costs, and fairness violations—the optimal control policy aims to minimize the cumulative cost while ensuring communication quality. During the backward dynamic programming process, the value function at each time step represents the expected future cost starting from that point. As time progresses toward the terminal horizon \(t \to T\), the remaining controllable duration diminishes, and the influence of current actions on future performance gradually weakens. Consequently, the value function naturally converges to zero at the terminal time \(T\), where no further cost can accrue and the terminal condition \(V(T) = 0\) is imposed as a boundary condition.

Figs. \ref{fig5} and \ref{fig6} show the SINR trajectories of the three UEs after DP iterations. The apparent lack of smoothness in these trajectories is a direct manifestation of two underlying factors. The principal factor is the stochastic nature of the channel model, which is driven by Lévy processes. Unlike the continuous sample paths generated by Brownian motion, the sample paths of a Lévy process are càdlàg (right-continuous with left limits) and possess inherent jump discontinuities. These jumps, which model impulsive physical phenomena, are the source of the abrupt variations observed in the SINR. A secondary factor is the numerical implementation of the optimal control policy via value iteration. The discretization of the action space leads to a quantized control policy, where the transmission power is adjusted in discrete steps. This results in piecewise constant control inputs, which contribute to the non-differentiable character of the resulting SINR trajectories. Nevertheless, the optimized power control strategy adapts in real time, effectively mitigating the impact of channel discontinuities. As the iteration progresses from 50 to 100 rounds, the SINR curves of all users become more balanced and less volatile. Despite the presence of heavy-tailed, non-Gaussian fading, the control policy demonstrates its robustness by effectively counteracting the majority of fluctuations and maintaining the user SINRs in a stable region around the predefined threshold of $r_{th}=1.5$. While occasional, sharp dips below the threshold can occur—an inherent characteristic of impulsive Lévy noise—the controller successfully maintains the target QoS for the vast majority of the operational time, showcasing its resilience.

\begin{figure}[t]
\centering
\subfloat[Noise Multiplier = 0.1]{\includegraphics[width=3.0in]{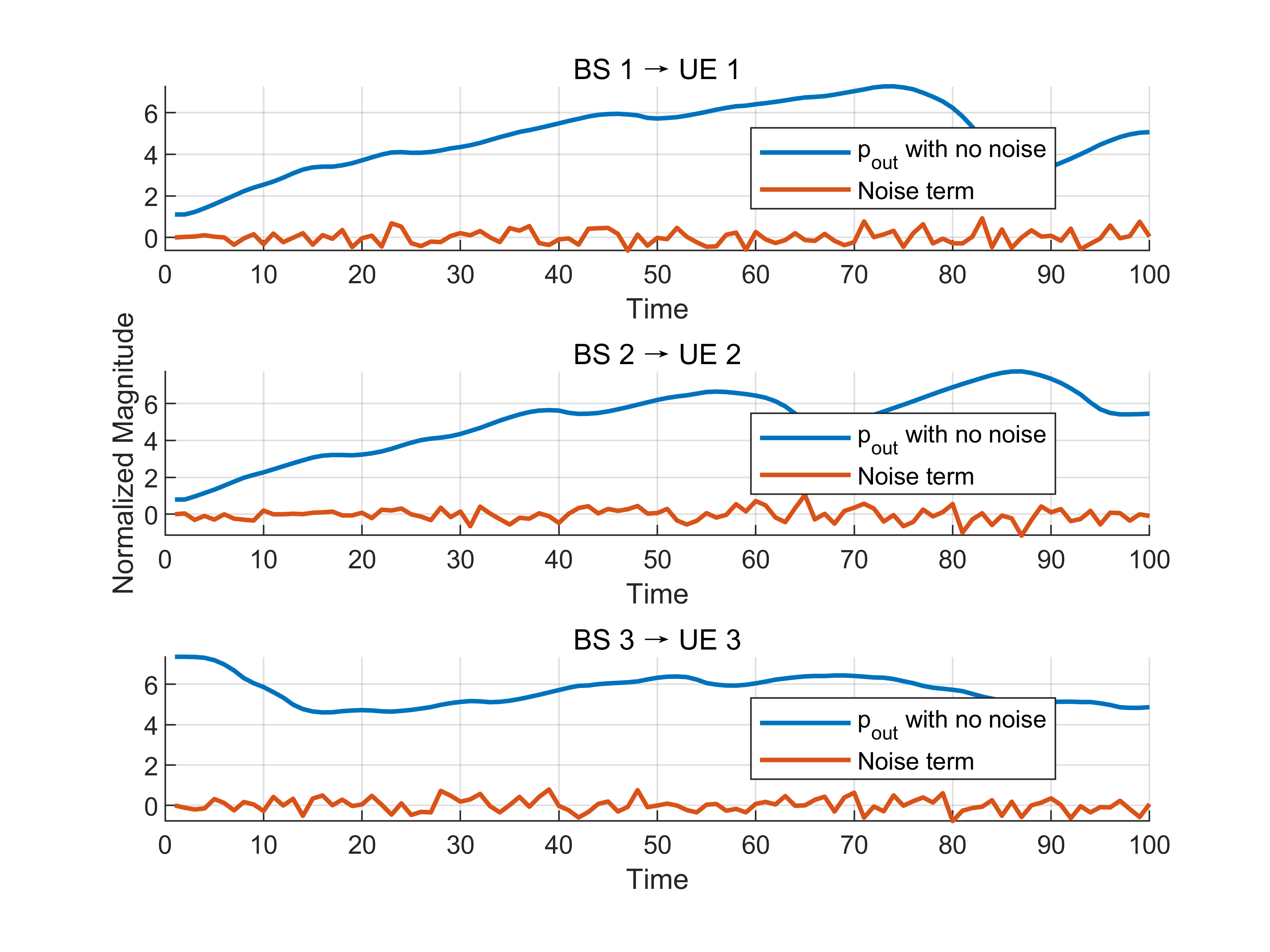}%
\label{fig7}}
\hfil
\subfloat[Noise Multiplier = 0.5]{\includegraphics[width=3.0in]{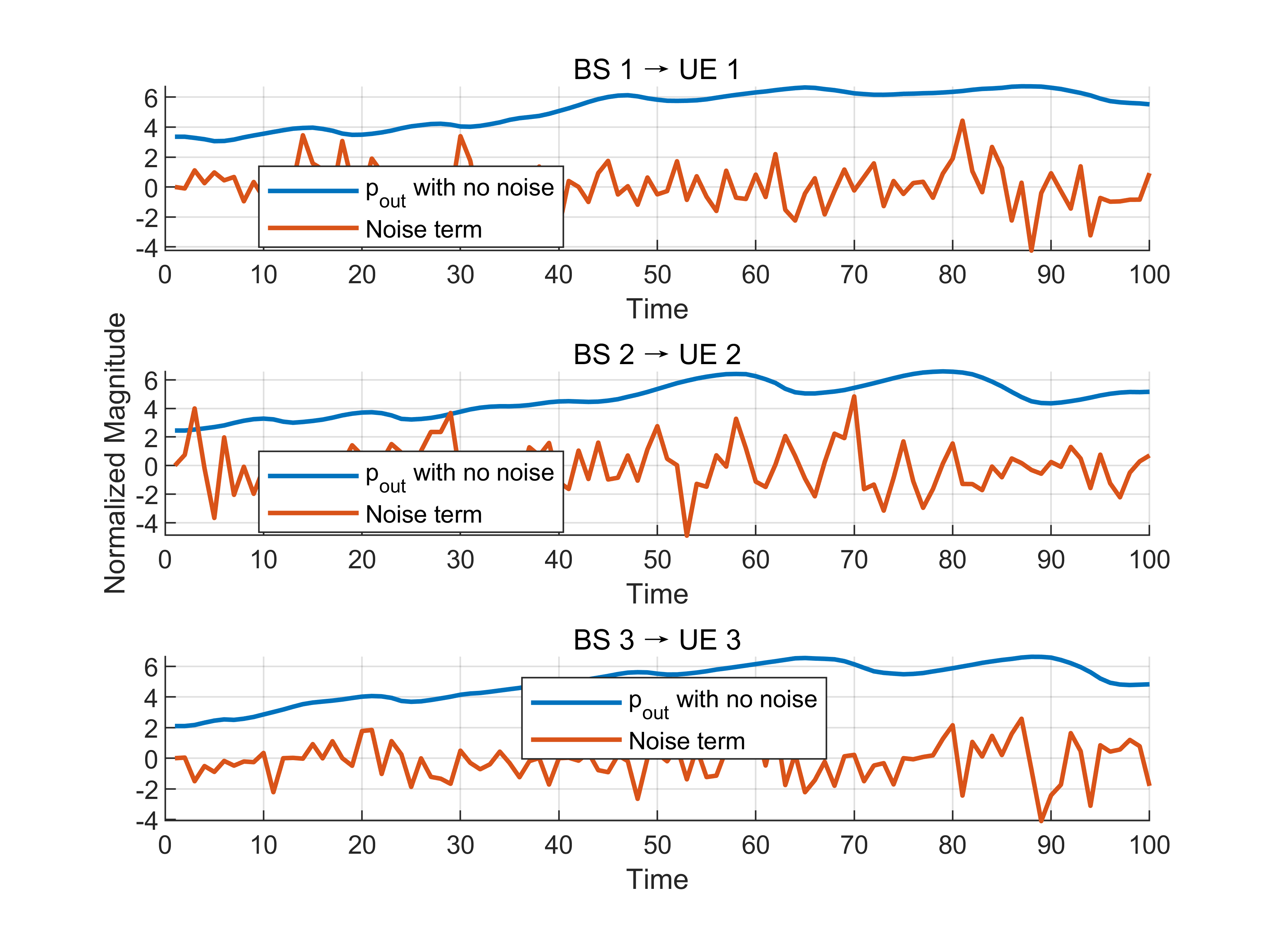}%
\label{fig8}}
\caption{Comparison between Received Power \(p_{i\ell}(t)\) and Stochastic Noise Term under Different Noise Intensities}
\label{Fig3}
\end{figure}

Fig. \ref{Fig3} compares the received power \(p_{i\ell}(t)\) and the corresponding Lévy-driven noise term under two noise intensities. When the noise scaling coefficient is 0.1, the stochastic fluctuations remain bounded and consistently smaller than the received signal power, enabling the controller to suppress impulsive interference and maintain robust system performance. However, when the coefficient increases to 0.5, the noise term exhibits large and unstable jumps, and the absolute value of the fluctuation amplitude exceeds the power. In this regime, the control strategy is overwhelmed by heavy-tailed disturbances, and \(p_{i\ell}(t)\) loses physical significance as a controllable quantity. This contrast highlights the necessity of constraining noise intensity in Lévy-driven environments to preserve optimization reliability and ensure stable power control.

\subsection{Comparative Analysis: Lévy vs. Gaussian Channels}
\begin{figure*}[htbp]
\centering
\subfloat[Initial]{\includegraphics[width=2.3in]{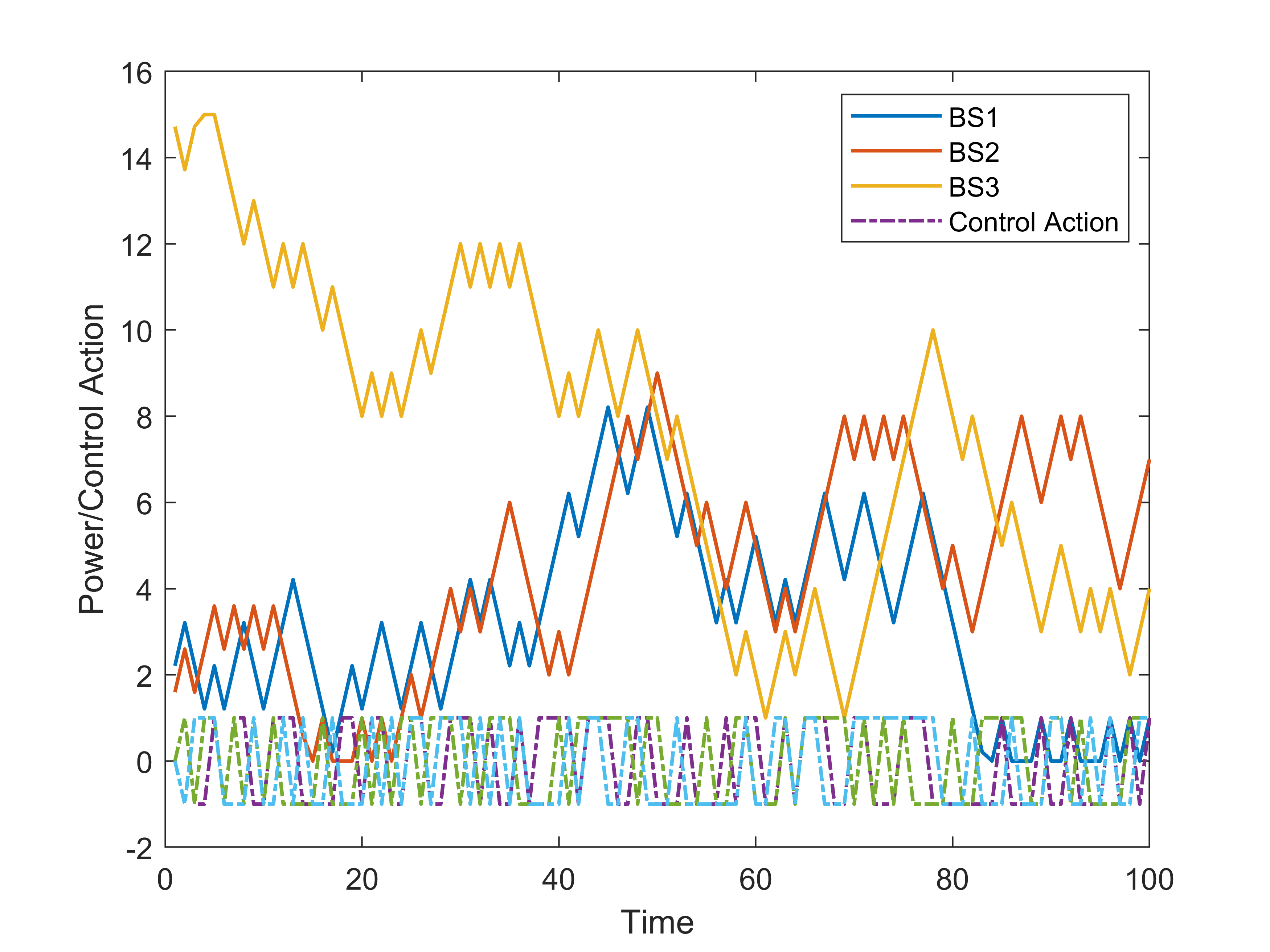}%
\label{fig9}}
\hfil
\subfloat[100 rounds]{\includegraphics[width=2.3in]{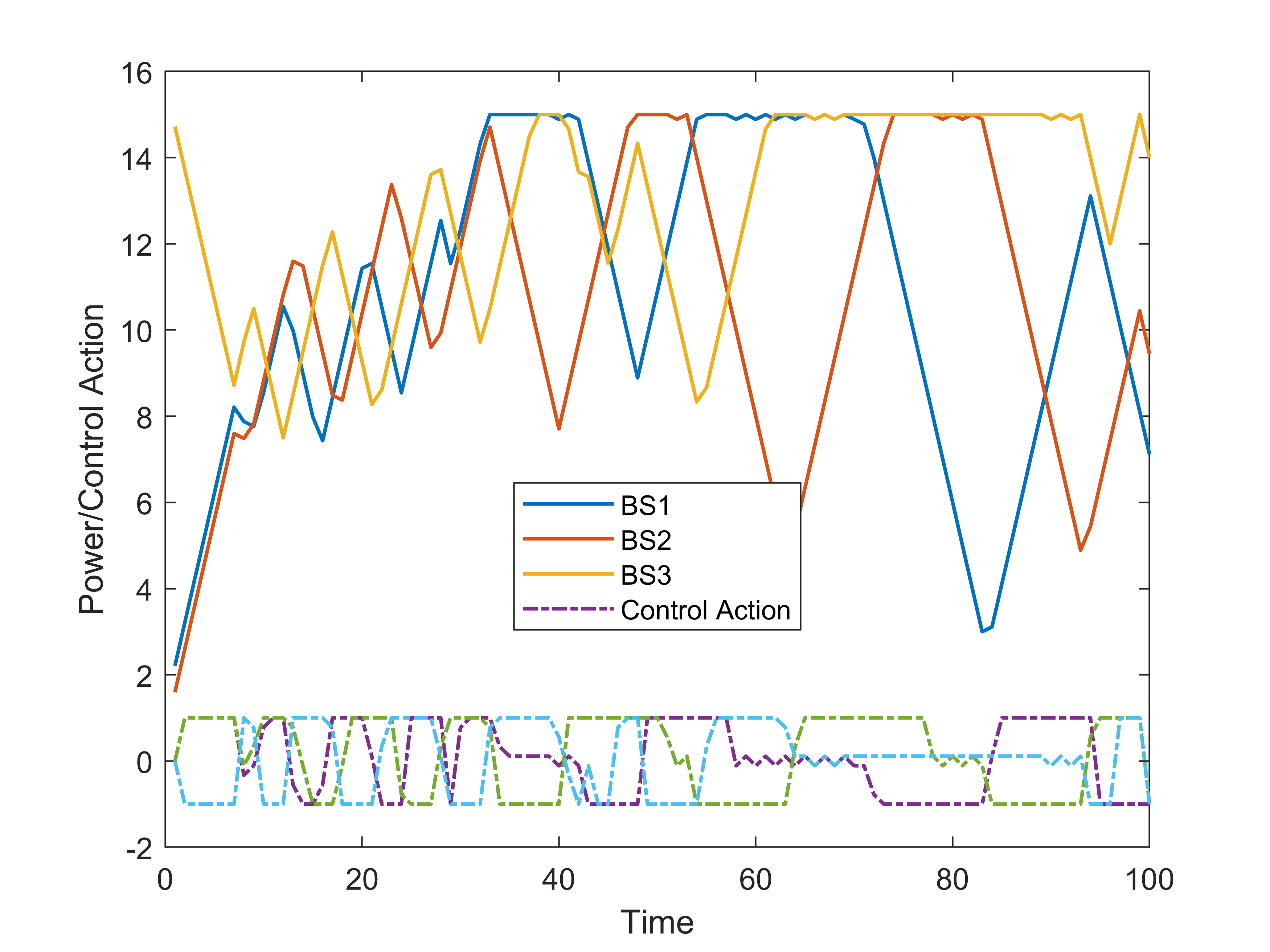}%
\label{fig10}}
\hfil
\subfloat[100 rounds]{\includegraphics[width=2.3in]{0Pin100.png}%
\label{fig11}}
\caption{Evolution of the BS Transmit Power and Control Actions after DP Iterations. (a) and (b) with $\alpha=2$; (c) with $\alpha=1.8$.}
\label{Fig4}
\end{figure*}

\begin{figure*}[htbp]
\centering
\subfloat[Value]{\includegraphics[width=2.3in]{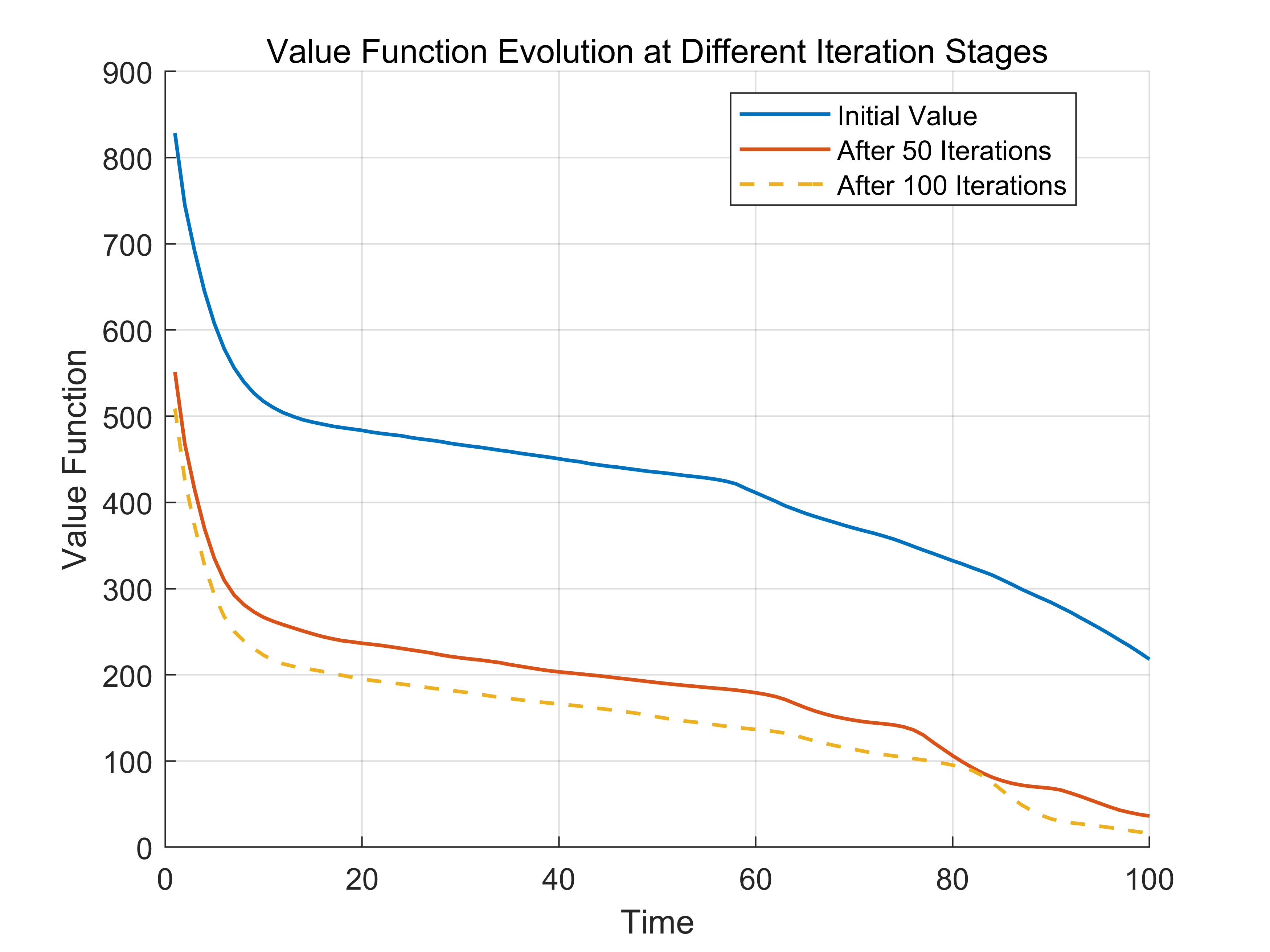}%
\label{fig12}}
\hfil
\subfloat[SINR: 100 rounds]{\includegraphics[width=2.3in]{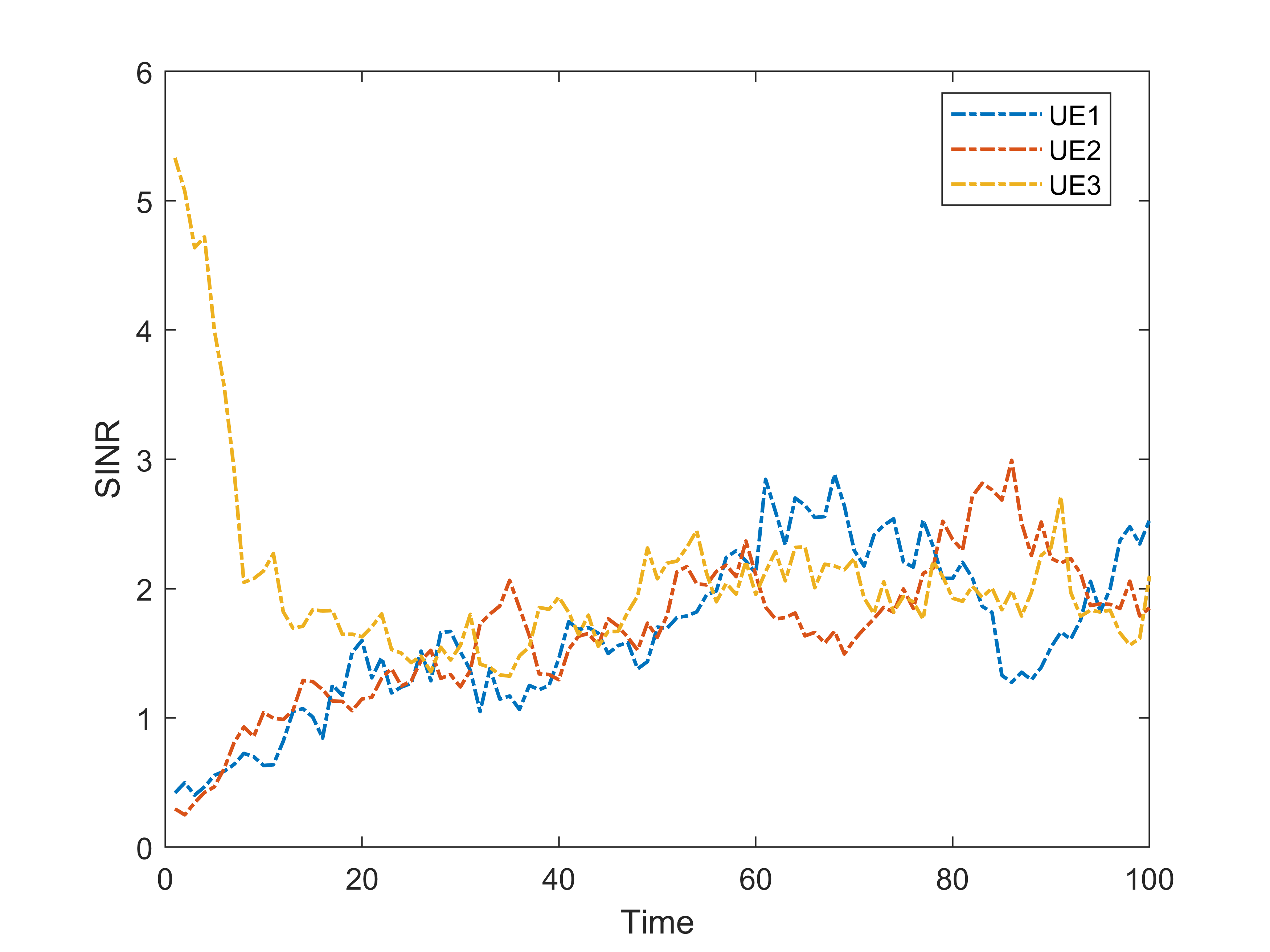}%
\label{fig13}}
\hfil
\subfloat[SINR: 100 rounds]{\includegraphics[width=2.3in]{0SINR100.png}%
\label{fig14}}
\caption{Evolution of the Value Function and SINR Trajectories after DP Iterations. (a) and (b) with $\alpha=2$; (c) with $\alpha=1.8$.}
\label{Fig5}
\end{figure*}
To isolate the contributions of our framework, we conducted a comparative analysis between the proposed S$\alpha$S Lévy-based controller ($\alpha$=1.8) and its conventional Gaussian limit ($\alpha$=2).

For this comparison, it is important to clarify the specific cost function minimized by the value iteration algorithm. The cost is composed of a running cost $L(t)$ and a terminal cost $V_T$. The running cost includes a quadratic penalty on power increases, a linear penalty for SINR falling below a threshold of $r_{\text{th}}=1.5$ (with weight $\varsigma=1$), and a fairness penalty based on SINR variance. The terminal cost includes similar terms and a reward for the final sum-rate. To ensure stable policies in both distinct environments, the sum-rate reward weight in the terminal cost was set to \emph{1} for the Gaussian case and \emph{0.1} for the Lévy case.

The converged power control policies reveal an adaptation to the noise statistics. Under Gaussian noise (Fig. \ref{fig10}), the control policy is characterized by relatively moderate power adjustments. In contrast, the policy under Lévy noise (Fig. \ref{fig11}) is notably more dynamic, a direct consequence of the heavy-tailed nature (infinite variance) of the S$\alpha$S process. To guard against the higher probability of high-magnitude, impulsive jumps, the controller must adopt a more aggressive strategy, making frequent power adjustments to maintain system stability in an environment with extreme statistical uncertainty.

The resulting SINR performance, depicted in Fig. \ref{Fig5}, demonstrates the advantage of the fractional framework in managing channels with Lévy properties. The trajectory under Gaussian noise is comparatively smooth (Fig. \ref{fig13}). The SINR in the Lévy environment (Fig. \ref{fig14}), however, reflects the jump discontinuities inherent to the driving process. The crucial demonstration of robustness is not in achieving a smooth path, but in the control action's ability to prevent these jumps from causing prolonged service outages. Despite being subjected to high-magnitude disturbances from the heavy tails of the S$\alpha$S distribution, the control policy works to guide the SINR back toward the target operating region after a disruption, showcasing its stabilizing effect against the non-Gaussian phenomena it was designed to handle.

\section{Conclusion}\label{S6}
This paper proposes a general optimal control framework for wireless systems operating under heavy-tailed and non-stationary fading, based on a novel wireless channel model driven by symmetric $\alpha$-stable Lévy processes. A fractional Hamilton-Jacobi-Bellman equation incorporating a spatial Riesz fractional operator is formulated, thereby enabling the modeling of non-local interactions and memory effects induced by impulsive dynamics. The existence and uniqueness of viscosity solutions are formally established, ensuring the theoretical soundness of the proposed formulation. In the numerical simulation, the fractional HJB–based power-control policy converges within 100 iterations to coordinated transmit-power profiles in a multi–base-station, multi-user downlink setting. This policy demonstrates robust performance under heavy-tailed, non-Gaussian interference by successfully maintaining the UE's SINR within a stable operating range around the 1.76 dB threshold, effectively mitigating the impact of large, Lévy-driven jumps.


\appendix

\section*{Appendix A: Proof of the existence of viscosity solutions}
Define the set of admissible viscosity subsolutions
\begin{align}
\mathcal{S} =& \left\{ w \in \USC([0, T] \times \mathbb{R}^d)  \bigg|  \underline{V} \leq w \leq \overline{V} , \text{$w$ is a viscosity}\right. \notag\\
& \text{subsolution of \eqref{eq1}},\Big. w(T, x) \leq g(x(T)) \Big\}.
\end{align}
Further, define the candidate solution as the upper envelope of subsolutions:
\begin{equation}\label{eq16}
V(t, x) = \sup_{w \in \mathcal{S}} w(t, x). \end{equation}
The goal is to show that $V$ is a viscosity solution of \eqref{eq1}.

\textbf{Step I.} 
Let $V^*$ be the upper semicontinuous envelope of $V$. Take $(t_0, x_0) \in(0, T) \times \mathbb{R}^d$, and a test function $\phi \in C^{1,2}\left((0, T] \times \mathbb{R}^d\right)$ such that $V^*-\phi$ attains a local maximum at $(t_0, x_0)$. 

By the definition of $u^*$ and the supremum in \eqref{eq16}, for any $\eta>0$, there exists $w_{\varepsilon} \in \mathcal{S}$ such that
\begin{equation} V^*(t_0, x_0) - \eta < w_\varepsilon(t_0, x_0) \leq V(t_0, x_0). \end{equation}
Given that $w_{\varepsilon} \leq V\leq V^*$, it follows that
\begin{equation}
    w_\epsilon(t, x)\leq V^*(t, x).
\end{equation}
Then, $w_{\varepsilon}-\phi$ attains a local maximum at $\left(t_0, x_0\right)$. Since $w_{\varepsilon}$ is a viscosity subsolution, it holds that
\begin{align} 
\partial_t \phi + H\left(\nabla_{x} \phi\right)+ \nabla^\alpha_\sigma[\phi] \leq \eta,\quad\text{at}~(t_0, x_0). \end{align}

Letting $\eta \rightarrow 0$, it can be concluded that $V$ is a viscosity subsolution.

\textbf{Step II.} Similarly, it is necessary to demonstrate that the lower semicontinuous envelope $V_{*}$ of $u$ is a viscosity supersolution.

Suppose, for contradiction, that there exists a test function $\phi$ such that $V_{*}-\phi$ attains a local maximum at $(t_0, x_0)$, and
\begin{align}
\partial_t \phi + H\left(\nabla_{x} \phi\right) + \nabla^\alpha_\sigma[\phi] < 0. \end{align}

Let us introduce the following function
\begin{equation} \tilde{V}(t, x) = \begin{cases} \max\{ V(t, x), \phi(t, x) \}, & \text{if } t \leq t_0, \\ V(t, x), & \text{if } t > t_0. \end{cases} \end{equation}

Given that $V_{*} - \phi$ attains a strict local minimum at $(t_0, x_0)$, the equality $\phi(t_0, x_0) = V_{*}(t_0, x_0)$ necessarily follows. Therefore, $\tilde{V}$ coincides with $V$ except possibly in a neighborhood of $(t_0, x_0)$. It is demonstrated that $\tilde{V} \in \mathcal{S}$: $\tilde{V}$ is upper semicontinuous, satisfies $\tilde{V}(T,x)(x) \leq V(T,x)$, and is a viscosity subsolution of \eqref{eq1}.

However, at $(t_0, x_0)$, the equality $\tilde{V}(t_0, x_0) = \phi(t_0, x_0) > V(t_0, x_0)$ holds, which contradicts the definition of $V$ as the supremum of subsolutions in $\mathcal{S}$. Hence, the initial hypothesis is not valid, and $V_{*}$ is a viscosity supersolution.

\textbf{Step III.} Finally, it is necessary to verify whether the terminal conditions are met, that is to say,  $V(T, x)=g(x(T))$.

Since all $w \in \mathcal{S}$ satisfy $w(T, x) \leq g(x(T))$, it is shown that
\begin{equation} V(T, x) = \sup_{w \in \mathcal{S}} w(T, x) \leq g(x(T)). \end{equation}

On the other hand, from the supersolution property of $V_{*}$, this necessarily leads to 
\begin{equation}
V(T, x) \geq g(x(T)). 
\end{equation}

Therefore, the conclusion is drawn that 
\begin{equation}
V(T, x) = g(x(T)). 
\end{equation}

Putting together the above results, it can be established that $V$ is a viscosity solution of the fractional HJB equation in \eqref{eq1} satisfying the terminal condition $g(x(T))$. Moreover, since $V$ is the supremum of bounded and uniformly continuous functions, it is itself bounded and uniformly continuous.

\section*{Appendix B: Proof of the uniqueness of viscosity solutions}
Assume, contrary to the claim, that
\begin{equation}
    W_2 := \sup_{(t, x) \in [0, T] \times \mathbb{R}^d} \left( V_1(t, x) - V_2(t, x) \right) > 0.
\end{equation}
Since $V_1$ is upper semicontinuous and $V_2$ is lower semicontinuous, the supremum $W_2$ is attained at some $(\hat{t}, \hat{x}) \in[0, T] \times \mathbb{R}^d$. Standard arguments (and the requirement that $V_1(T, \cdot) \leq V_2(T, \cdot)$ ) imply that $\hat{t} \neq T$, since otherwise it would contradict the terminal condition. Hence, $\hat{t} \in[0, T)$.

The penalized function is defined by
\begin{align}
\Phi_\varepsilon(t, x, s, y) = V_1(t, x) - V_2(s, y) - \frac{| t - s |^2+|| x - y ||^2}{2 \varepsilon^2}, 
\end{align}
where $\varepsilon>0$ is a small parameter.

According to the upper and lower semicontinuity, there exists $(t_\varepsilon,x_\varepsilon, s_\varepsilon, y_\varepsilon)$ that maximizes $\Phi_\varepsilon$. As $\varepsilon\to 0$,
\begin{equation}
t_\varepsilon, s_\varepsilon \to \hat{t}, \quad x_\varepsilon, y_\varepsilon \to \hat{x},
\end{equation}
and
\begin{equation}
V_1(t_\varepsilon, x_\varepsilon) - V_2(s_\varepsilon, y_\varepsilon) \to W_2. 
\end{equation}

The following test functions are considered for further analysis:
\begin{align}
\phi^1(t, x) &= \frac{| t - s_\varepsilon |^2+| x - y_\varepsilon |^2}{2 \varepsilon^2} , \notag\\ \phi^2(s, y) &= -\frac{| t_\varepsilon - s |^2+| x_\varepsilon - y |^2}{2 \varepsilon^2}.
\end{align}

$V_1\left(t, x\right)-\phi^1\left(t, x\right)$ and $V_2\left(s,y\right)+\phi^2\left(s, y\right)$ attain local maxima at $\left(t_{\varepsilon}, x_{\varepsilon}\right)$ and $\left(s_{\varepsilon}, y_{\varepsilon}\right)$, respectively. 

By Def.~\ref{def1}, $V_1$ is a viscosity subsolution:
\begin{align}
\partial_t \phi^1 + H\left( \nabla_{x} \phi^1\right) + \nabla^\alpha_\sigma[\phi^1] \leq 0, \quad \text{at}~ (t_\varepsilon, x_\varepsilon). 
\end{align}
Similarly, by Def.~\ref{def2},  $V_2$ is a viscosity supersolution:
\begin{align}
\partial_s (-\phi^2) + H\left(\nabla_{y} (-\phi^2)\right) + \nabla^\alpha_\sigma[-\phi^2] \geq 0 , \quad \text{at}~ (s_\varepsilon, y_\varepsilon). 
\end{align}

Next, it is worth noting that 
\begin{equation}
\partial_t \phi^1(t_\varepsilon, x_\varepsilon) = \frac{t_\varepsilon - s_\varepsilon}{\varepsilon^2}, \nabla_{x} \phi^1(t_\varepsilon, x_\varepsilon) = \frac{x_\varepsilon - y_\varepsilon}{\varepsilon^2}.
\end{equation}
Analogous expressions hold for $\partial_s(-\phi^2)$ and $\nabla_y(-\phi^2)$. Using the Lipschitz property assumption \textbf{(A2)}, we have
\begin{align}  H\left(x_\varepsilon, \nabla_{x} \phi^1\right) - H\left(y_\varepsilon, \nabla_{y} (-\phi^2)\right) \rightarrow 0,\quad as ~\varepsilon\rightarrow 0.\end{align}

Since both $\phi^1$ and $-\phi^2$ are of ``penalized quadratic'' form around $\left(t_{\varepsilon}, x_{\varepsilon}\right)$ and $\left(s_{\varepsilon}, y_{\varepsilon}\right)$, the Riesz fractional derivatives of such functions are zero or bounded in sufficiently small neighborhoods, provided that they are polynomial-like away from those neighborhoods. By carefully choosing compactly supported (or sufficiently decaying) extensions of $\phi^1,-\phi^2$, one can show that this difference is small and vanishes as $\varepsilon \rightarrow$ 0. 

We note that the fractional terms satisfy
\begin{equation}
\left|\nabla_\sigma^\alpha\left[\phi^1\right]-\nabla_\sigma^\alpha\left[-\phi^2\right]\right| \rightarrow 0, \quad as~\varepsilon \rightarrow 0.
\end{equation}

Given the terminal condition $V_1(T, x) \leq V_2(T, x)$, a standard contradiction arises if $W_2>0$ is attained in the interior $\hat{t}<T$. Therefore, the only consistent resolution is $W_2 \leq 0$, which implies $V_1 \leq V_2$ on $[0, T] \times \mathbb{R}^d$.

\newpage

\end{document}